\documentclass[onefignum,onetabnum]{siamonline220329}

\usepackage{cite}
\usepackage[utf8]{inputenc} % allow utf-8 input
\usepackage[T1]{fontenc}    % use 8-bit T1 fonts
\usepackage{url}            % simple URL typesetting
\usepackage{booktabs}       % professional-quality tables
\usepackage{nicefrac}       % compact symbols for 1/2, etc.
\usepackage{microtype}      % microtypography
\usepackage{bbold}
\usepackage{mathtools}
\usepackage{enumitem}
\usepackage{amssymb,amsmath,amsfonts,amsopn}
\usepackage{mathabx}
\usepackage{changepage}
\usepackage{color,graphicx,float}
\usepackage{bm}
\usepackage{alltt,listings,tikz}
\usepackage[all]{xy}
\usepackage{hyperref}
\usepackage{dsfont}
\usepgflibrary{shadings}
\usepackage[english]{babel}
\usepackage{graphicx}
\usepackage{epstopdf}
\usepackage{algorithmic}

\ifpdf
  \DeclareGraphicsExtensions{.eps,.pdf,.png,.jpg}
\else
  \DeclareGraphicsExtensions{.eps}
\fi

\usepackage{enumitem}
\setlist[enumerate]{leftmargin=.5in}
\setlist[itemize]{leftmargin=.5in}

\newsiamremark{remark}{Remark}
\headers{Non-Markovian models of opinion dynamics on temporal networks}{W. Chu
 and M. A. Porter}

\newcommand{\add}[1]{\textcolor{black}{{#1}}}

\title{Non-Markovian models of opinion dynamics on temporal networks}
\author{Weiqi Chu\thanks{Department of Mathematics, University of California,  Los Angeles
 }
 \and Mason A. Porter\thanks{Department of Mathematics, University of California,  Los Angeles and Santa Fe Institute
 }
}

\ifpdf
\hypersetup{
  pdftitle={Non-Markovian models of opinion dynamics on temporal networks},
  pdfauthor={W. Chu, and M. A. Porter}
}
\fi

\begin{document}

\maketitle

\begin{abstract}
Traditional models of opinion dynamics, in which the nodes of a network change their opinions based on their interactions with neighboring nodes, consider how opinions evolve either on time-independent networks or on temporal networks with edges that follow Poisson statistics. Most such models are Markovian. However, in many real-life networks, interactions between individuals (and hence the edges of a network) follow non-Poisson processes and thus yield dynamics with memory-dependent effects. In this paper, we model opinion dynamics in which the entities of a temporal network interact and change their opinions via random social interactions.
When the edges have non-Poisson interevent statistics, the corresponding opinion models are have non-Markovian dynamics. We derive an opinion model that is governed by an arbitrary waiting-time distribution (WTD) and illustrate a variety of induced opinion models from common WTDs (including Dirac delta distributions, exponential distributions, and heavy-tailed distributions). We analyze the convergence to consensus of these models and prove that homogeneous memory-dependent models of opinion dynamics in our framework always converge to the same steady state regardless of the WTD. We also conduct a numerical investigation of the effects of waiting-time distributions on both transient dynamics and steady states. We observe that models that are induced by heavy-tailed WTDs converge to a steady state more slowly than those with light tails (or with compact support) and that entities with larger waiting times exert a larger influence on the mean opinion at steady state.

\end{abstract}



%%%%%%

%%%%%%

\section{Introduction}

The structure of networks has a major influence on the dynamics of complex systems of interacting entities in social, economic, information, biological, and physical systems \cite{newman2018networks}. In a network, entities interact via edges, which encode ties with time-dependent strengths.
To model a networked system, it is important both to account for the time-dependence of edges and to examine the effects of (both time-independent and time-dependent) network structures on dynamical processes \cite{porter2016}, such as opinion formation \cite{noor2020}, the spread of infectious diseases \cite{kiss2017}, and e-mail communication \cite{eckmann2004entropy}. 

When modeling real-life networks, it is convenient to assume that edges, which encode events between humans or other entities, appear randomly in a way that is well-captured by a Poisson process. This assumption results in time-dependent networks (i.e., so-called ``temporal networks'') with memoryless stochastic effects \cite{masuda2016guide}. It overlooks the non-Markovian and nonstationary nature of many systems \cite{barabasi2005origin}, such as e-mail traffic, online communication, and others. 
\add{To incorporate memory effects and to model dynamics with bursty and heavy-tailed interevent times \cite{karsai2011small,scholtes2014causality}, it is important to consider temporal networks with edges that appear according to stochastic processes other than Poisson processes.
}

Many dynamical processes on networks are non-Markovian \cite{feng2019equivalence,gleeson2014}, which introduces nontrivial memory-dependence into their dynamics. For example, in a social contagion, entities typically require multiple sources of influence (i.e., social reinforcement) to adopt some idea or behavior \cite{sune-yy2018}. Such cumulative effects occur in the adoption of social norms and technologies \cite{bandura1963influence,centola2010spread}.
There have been a variety of efforts to incorporate memory effects into dynamical processes on networks. Such studies include generalizations of voter models \cite{chen2020non,takaguchi2011voter,baron2022analytical}, compartmental models of disease spread \cite{starnini2017equivalence,van2013non,kiss2015generalization,feng2019equivalence}, social-contagion models \cite{wang2015dynamics}, and random walks~\cite{lambiotte2015effect}.

Most research on incorporating memory effects into dynamical process on networks has focused on binary-state models, in which node states can take one of two values (e.g., susceptible or infected), but it is also important to examine memory effects in models in which nodes can take continuous values. 
Most such research incorporates memory-dependence directly into network structure through time-dependent edge weights and overlooks the effects of the previous node states. For example, Sugishita et al. considered an opinion model on ``tie-decay networks'' \cite{sugishita2021opinion}, in which interactions between entities are time-dependent and result in ties whose strengths increase instantaneously when an interaction occurs and decay exponentially between interactions. However, simply placing a dynamical process on a tie-decay network accounts only for the states during the most recent interaction; it ignores how states changed with time to attain their current values. By contrast, we seek a model formulation that accounts for the complete {history} of states of a dynamical process.

In addition to the just-discussed temporal influence of networks, dynamical processes on networks are also impacted significantly by network architecture \cite{porter2016}. For example, Meng et al.~\cite{meng2018opinion} studied how time-independent network structures affect the steady state and the convergence properties of bounded-confidence models (BCMs) of opinion dynamics. Sood et al. \cite{sood2005voter} investigated the relationship between a network's degree distribution and the convergence time of a voter model on that network. Delvenne et al. \cite{delvenne2015} compared the effects of the structural and temporal features of networks on dynamical processes. In particular, they examined when a diffusion process is affected more strongly by network architecture or by a network's temporal features. Such investigations emphasize the importance of considering generic network structures when studying both transient and long-time qualitative behaviors of dynamical processes.

In the present study of non-Markovian opinion dynamics on networks, we consider arbitrary weighted networks and allow the waiting time between events to follow an arbitrary probability distribution, which is known as a ``waiting-time distribution'' (WTD). When the edges of a network satisfy {Poisson} statistics, the times between events follow independent exponential distributions \cite{medhi1975waiting} and thereby lead to memoryless models in which the dynamics depends only on a network's present state. Instead of using such a restrictive setting, we consider a generic WTD and use known results about interevent times \cite{kivela2015estimating} to systematically construct memory-dependent models of opinion dynamics.
This setting allows us to study models that capture time-dependent interactions between entities and naturally incorporate memory effects.

Our paper proceeds as follows. In Section \ref{sec: models}, we propose a family of memory-dependent 
models of opinion dynamics on temporal networks in which the social interactions between two entities are determined by arbitrary WTDs.
We illustrate the corresponding opinion models for several examples of both discrete and continuous WTDs. We prove that ``homogeneous models'' (in which all nodes follow the same WTD) of this type converge to the same steady state regardless of the WTD, and we give conditions for consensus in both homogeneous models and ``heterogeneous models'' (in which nodes can follow different WTDs).
In Section \ref{sec: numerics}, we {examine} our memory-dependent opinion
models {on} three types of graphs and {investigate} how WTDs affect the overall dynamics both transiently and at steady state. 
We conclude in Section \ref{sec: summary}. 
Our code is available at \url{https://bitbucket.org/chuwq/non-markovian-models-of-opinion-dynamics-on-temporal-networks/src/main/}.

%%%%%%%%%%%%%%%%%%%%%%%%%%%%%%%%%%%%%%%%%%%%%%%%%%%%%%%%%%%%%%%%%%%%%%%%%%%%%%%%%%%%%%%%%%%%%%%%%%%%%%%%%%%%%%%

\section{Opinion models that are induced by waiting-time distributions (WTDs)} 
\label{sec: models}

\add{Let $G$ be a weighted and directed graph with $N$ nodes. We represent this graph using an the adjacency matrix $A$. Entry $A_{ij}$ of this matrix gives the weight of the edge from node $i$ to node $j$; it encodes the interaction strength from entity $i$ to entity $j$.\footnote{To consider an unweighted graph, we let $A_{ij} = 1$ when there is an edge and $A_{ij} = 0$ when there is not an edge. In an undirected graph, $A_{ij} = A_{ji}$ for all $i$ and $j$.}
We assume that the entries of $A$ are nonnegative real numbers and each row sum of $A$ is larger than $0$. The \emph{row-normalized adjacency matrix} $\widetilde{A}$ has entries $\widetilde{A}_{ij} = A_{ij}/\sum_{j=1}^N A_{ij}$.
Entity $i$ has a time-dependent continuous-valued opinion $X_i(t)$ and an internal clock $\tau_i$, which {indicates its waiting time.}
Each entity maintains its opinion until there is an event (which is determined by $\tau_i$). 
The waiting time $\tau_i$ between two consecutive events of entity $i$ follows a WTD $T_i(\tau)$. When an event occurs, entity $i$ adopts the opinion of an adjacent node\footnote{When $A_{ij} > 0$, node $j$ is adjacent to node $i$. When $A_{ij} = 0$, nodes $i$ and $j$ are not adjacent.} $j$ with probability $\widetilde{A}_{ij}$ and it resets its waiting time $\tau_i$ to $0$. 
When $A_{ii} > 0$, entity $i$ is adjacent to itself, so it can choose itself (with probability $\widetilde{A}_{ii}$) when selecting a node from which to adopt an opinion. If this occurs, entity $i$ keeps its current opinion after the event and it resets the waiting time $\tau_i$ to $0$.
In Figure \ref{fig: diagram}, we give an example of a two-node graph and illustrate how its entities update their opinions.
}

\begin{figure}[htp!]
    \centering
    \includegraphics[width=0.8\textwidth]{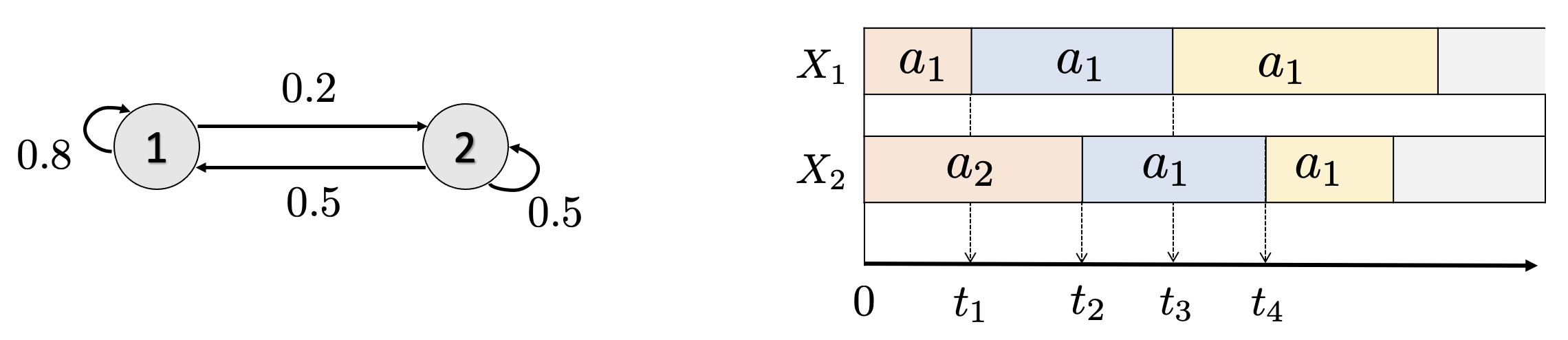}
    \caption{\add{(Left) A $2$-node weighted graph with self-edges and (right) the trajectories of associated opinion trajectories ($X_1$ and $X_2$). Initially, the two entities of the graph have opinions $a_1$ and $a_2$. At time $t_1$, entity $1$ experiences an event; it adopts opinion $a_1$ (which is the opinion of entity 1) with probability $\widetilde{A}_{11} = 0.8$, and it adopts the opinion $a_2$ (which is the opinion of entity 2) with probability $\widetilde{A}_{12} = 0.2$. In the depicted scenario, entity $1$ adopts opinion $a_1$ and holds that opinion until its next event, which occurs at time $t_3$. Entity $2$ holds opinion $a_2$ until its first event occurs at time $t_2$. Entity $2$ chooses the opinion $a_1$ of entity $1$ at time $t_2$; this results in a consensus state, in which both entities hold the same opinion $a_1$.}
    }
    \label{fig: diagram}
\end{figure}

\add{The internal clock $\tau_i$ and WTD $T_i(\tau)$ determine when entity $i$ can update its opinion, and the normalized weight $\widetilde{A}_{ij}$ gives the probability that entity $i$ adopts the opinion of entity $j$. 
We assume that the events of different entities occur independently. If entity $i$ adopts opinion $X_j(t)$ at time $t$, entity $j$ still updates its opinion according to its internal clock $\tau_j$ without noticing that entity $i$ has adopted its opinion.} Therefore, we consider unidirectional interactions between entities. 
Such interactions arise in many social and biological systems \cite{guttal2010social}. For example, followers of a social-media account can update their opinions by merely reading posts without commenting or otherwise actively communicating with that account.
As in \cite{hoffmann2012generalized}, our graphs $G$ are temporal networks because of the WTDs of the nodes.

In the following subsections, we study models of opinion dynamics that are induced by the WTDs of nodes on temporal networks. In Section \ref{sec: general_process}, we derive a master equation for the time-dependent opinions of $N$ nodes, which can have different WTDs.
In Section \ref{sec: derivation}, we examine models of opinion dynamics that use several well-known WTDs (Dirac delta distributions, exponential distributions, and heavy-tailed distributions) for the interevent times. In Section \ref{sec: analysis}, we examine convergence and consensus in our non-Markovian opinion models for both homogeneous and heterogeneous systems.

%%%%%

\subsection{Master equations for opinion dynamics with arbitrary WTDs} \label{sec: general_process}

Let $f_i(x,t)$ be the probability density function (PDF) of $X_i(t)$ on an opinion space $\Omega$, and let $q_i^{(k)}(x,t)$ be a \add{PDF} on $\Omega\times \mathbb{R}_{\ge 0}$. The PDF $q_i^{(k)}(x,t)$ governs the probability that entity $i$ adopts opinion $x$ at time $t$ in its $k$th event. For each entity 
$i$, we have
\begin{align}
    \int_{\Omega} f_i(x,t)\,dx &= 1 \text{  \, for any \, } t\ge0\,, \notag \\  
    \int_0^{\infty}\!\!\int_{\Omega} q_i^{(k)}(x,t)\,dx\,dt &= 1 \text{  \, for any \, } k\in \mathbb{N} = \{0, 1, \ldots \} \,.
\end{align}
\add{The $0$th event of each entity occurs at time $0$. At time $0$, entity $i$ updates its opinion according to the PDF $f_i(x,0)$ and sets its internal waiting-time clock $\tau_i$ to $0$.}
Suppose that entity $i$ holds opinion $x$ at time $t > 0$. This opinion arises from an earlier event at some time $t' \in [0, t)$ with no event of entity $i$ between times $t'$ and $t$. 
In mathematical terms, 
\begin{equation} \label{eq: eqfi}
    f_i(x,t) = \sum_{k=0}^{\infty}\int_0^t \phi_i(t-t')q^{(k)}_i(x,t')\,dt'\,,
\end{equation}
where 
\begin{equation}\label{eq: phi-eq}
    \phi_i(t) = 1- \int_0^t T_{i}(t')\,dt'
\end{equation}
is the probability that an event of entity $i$ does occur after waiting for time $t$. The function $\phi_i(t)$ is the survival function with respect to $T_i$. 
To iterate between two consecutive events of the same entity $i$, 
we write
\begin{equation} \label{eq: eqqi}
\begin{aligned}
    q_i^{(k+1)}(x,t) &= \sum_j \left[ \int_0^t \int_{\Omega} \add{q_i^{(k)}}(y,t')T_{i}(t-t')\,dy\,dt'\right] \widetilde{A}_{ij}f^-_j(x,t)\,, \quad k\in \mathbb{N}\,,
    \\
    q_i^{(0)}(x,t) &= f_i(x,0)\delta(t)\,.
\end{aligned}
\end{equation}
The density $f^-_j(x,t)$ is the limit of $f_j(x,\tau)$ as $\tau \rightarrow t^-$ (i.e., $f^-_j(x,t)=\lim_{\tau\rightarrow t^-}f_j(x,\tau)$). 
If entity $i$ has an event at time $t$ and entity $i$ adopts the state of entity $j$, we change $i$'s opinion $X_i$ to $X_j(t^-)$, which is the opinion of entity $j$ right before a possible event of entity $j$ at time $t$. When $T_{i}(t)$ does not possess a point mass (i.e., $T_{i}(t)$ does not have a Dirac delta measure), $f^-$ and $f$ are the same because there is $0$ probability that two events occur simultaneously.
When $T_{i}(t)$ possesses a point mass (i.e., $T_{i}(t)$ has a positive probability at one or more isolated points), we need to distinguish between $X_i(t^-)$ and $X_i(t)$ to avoid ambiguity in situations when multiple events occur simultaneously.

Let $\bar{q}^{(k)}_i(t) = \int_{\Omega}q_i^{(k)}(x,t)\,dx$ be the PDF on $\mathbb{R}_{\ge0}$ that the $k$th event of entity $i$ occurs at time $t$. 
Using Equation \eqref{eq: eqqi}, we obtain
\begin{equation} \label{eq: barqi}
\begin{aligned}
    \bar{q}_i^{(k+1)}(t) &= \int_0^t \bar{q}_i^{(k)}(t')T_{i}(t-t')\,dt'\,, \quad k\in \mathbb{N} \,,    \\
    \bar{q}_i^{(0)}(t) &= \delta(t)\,.
\end{aligned}
\end{equation}
Combining Equations (\ref{eq: eqfi}, \ref{eq: eqqi}, \ref{eq: barqi}) yields the governing equation for the probability densities of the opinions:
\begin{equation} \label{eq: eqf}
\begin{aligned}
    f_i(x,t) = \sum_{j} \phi_i(t) \star \left[\widetilde{A}_{ij}\theta_{i} (t)f^-_j(x,t)\right] + \phi_i(t)f_i(x,0)  \,,
\end{aligned}
\end{equation}
where $\star$ denotes time convolution 
\add{and 
\begin{equation}\label{eq: thetai}
\theta_{i} = \sum_{k=0}^\infty \bar{q}_i^{(k)}\star T_{i}\,.
\end{equation}
Let the hat $\hat{\cdot}$ denote a Laplace transform. Using \eqref{eq: barqi} and \eqref{eq: thetai}, we obtain that} the Laplace transforms $\hat{\theta}_i$ and $\hat{T}_i$ satisfy
\begin{equation}\label{eq: theta_i}
    \hat{\theta}_i = \left( 1-\hat{T}_i\right)^{-1}\hat{T}_{i}\,.
\end{equation}
Let $x_i(t)=\int_{\Omega}xf_i(x,t)\,dx$ be the expectation of $X_i(t)$. A direct computation from Equation \eqref{eq: eqf} shows that $x_i(t)$ satisfies the integral equation
\begin{equation} \label{eq: model-discrete}
	 x_i(t) = \sum_j \widetilde{A}_{ij}\left[\phi_i\star \left( \theta_{i} x^-_j\right)\right](t) + \phi_i(t)x_i(0) \,,
\end{equation}
where $\phi_i$ and $\theta_i$ are defined in \eqref{eq: phi-eq} and \eqref{eq: theta_i}, respectively. \add{Equation \eqref{eq: model-discrete} gives a family of \emph{memory-dependent opinion models} that are induced by arbitrary WTDs $T_i(t)$.
When $f_i(x,t)$ is continuous with respect to $t$, we take the Laplace transform of Equation \eqref{eq: model-discrete} and obtain
\begin{equation} \label{eq: hatxi}
\hat{x}_i = \sum_j \widetilde{A}_{ij} \hat{\phi}_i \widehat{ \theta_{i} x^-_j} + \hat{\phi}_ix_i(0)\,.
\end{equation}
We use the equalities
\begin{equation}
    \hat{\phi}_i(s) = \frac1s\left(1-\hat{T}_i(s)\right) \quad \text{and}  \quad  x_i(0) = s\hat{x}_i(s) - \hat{\dot{x}}_i(s)\,,
\end{equation}
and rewrite \eqref{eq: hatxi} as
\begin{equation}  \label{eq: inverse}
	\hat{\dot{x}}_i = \sum_j \widetilde{A}_{ij} \widehat{\theta_{i} x^-_j} -  \frac{s\hat{T}_i}{1-\hat{T}_i} \hat{x}_i\,.
\end{equation}
By taking the inverse Laplace transform of Equation \eqref{eq: inverse}, we obtain a set of coupled ordinary differential equations:}
\begin{equation}\label{eq: model-continuous}
    \dot{x}_i(t) = \sum_j \widetilde{A}_{ij}\theta_{i} (t)x_j(t) - \left[\chi_i \star x_i\right](t)\,,
\end{equation}
where the Laplace transform of $\chi_i$ satisfies $\hat{\chi}_i(s) =  s \hat{T}_i(s) \left(1-\hat{T}_i(s)\right)^{-1}$. 
{Equation \eqref{eq: model-continuous} gives a family of memory-dependent opinion models that are induced by continuous-time WTDs.}

We say that the opinion models \eqref{eq: model-discrete} are \emph{homogeneous} if all nodes (i.e., entities) of a network have the same WTD $T(t)$, which in turn leads to homogeneous survival functions $\phi(t)$ and $\theta(t)$; otherwise, we say that the opinion models \eqref{eq: model-discrete} are \emph{heterogeneous}. 
For the homogeneous case, we can simplify the memory-dependent opinion models \eqref{eq: model-discrete} and write
\begin{equation} \label{eq: homo}
	 x_i(t) = \sum_j \widetilde{A}_{ij}\left[\phi \star \left( \psi x^-_j\right)\right](t) + \phi(t)\sum_j \left(\delta_{ij}- \widetilde{A}_{ij}\right)x_j(0)  \,,
\end{equation}
where $\delta_{ij}$ is the Kronecker delta function and $\psi(t)$ is related to the WTD $T(t)$ by $\hat{\psi}(s) = (1-\hat{T}(s))^{-1}$. 
From a direct computation, we see for all $t \ge 0$ that
\begin{equation} \label{eq: psi-prop}
    \psi(t) = \theta(t) + \delta(t)\,, \quad \left(\phi\star\psi\right)(t) = 1\,, \quad \psi(t)\ge 0\,,
\end{equation}
where $\delta(t)$ is the Dirac delta function. 

%%%

\subsection{Models of opinion dynamics that are induced by common WTDs} \label{sec: derivation}

We now examine opinion models \eqref{eq: model-discrete} that are induced by several common WTDs, including both discrete-time and continuous-time distributions. Some WTDs, such as the exponential distribution, yield a Markovian dynamical process with a time discretization that matches the DeGroot model of opinion dynamics \cite{degroot1974reaching}. Other WTDs, such as the gamma distribution, yield non-Markovian dynamics; in these systems, the opinion state depends on the entire history of all opinion values. We also examine models that arise from heavy-tailed WTDs and study approximations of them using a sum of Dirac delta measures (when we do not have explicit formulas for the inverse Laplace transforms).

%%%

\subsubsection{Dirac delta WTD}

Consider a situation in which events occur after entities wait for a fixed amount of time. 
That is, the WTD of each node  
is the Dirac delta distribution $T_{i}(t) = \delta(t - \triangle_i)$, which yields
\begin{equation}
    \hat{T}_{i}(s) = e^{-\triangle_is}\,, \quad 
    \theta_{i} (t) = \sum_{k=1}^{\infty}\delta(t-k\triangle_i)\,, \quad
    \phi_i(t) = \mathbb{1}_{[0,\triangle_i)}(t)\,.
\end{equation}
In Equation \eqref{eq: model-discrete}, $x_j^-(t)$ is the opinion right before entity $j$ changes its opinion at time $t$. Because entity $j$ updates its opinion  
after waiting for exactly time $\triangle_j$, we have $x_j^-(t)=x_j(t-\triangle_j)$. This yields the opinion model 
\begin{equation} \label{eq: delta-model-i}
    \begin{aligned}
        x_i(t) &= x_i(0)\,, \quad && t \in [0,\triangle_i)\,, \\
        x_i(t) &= \sum_{j} \widetilde{A}_{ij} x_j(t-\triangle_j)\,, \quad && t\in [\triangle_i,\infty)\,.
    \end{aligned}
\end{equation}

Suppose that all entities wait for the same amount of time $\triangle_i = 1$ before updating their opinions (i.e., the model is homogeneous). We can then write \eqref{eq: delta-model-i} in matrix form as the discrete dynamical system
\begin{equation} \label{eq: dirac-model}
    x(n+1) = Px(n)\,,\quad n \in \mathbb{N}\,,
\end{equation}
where $x(n) \in \mathbb{R}^N$ is the vector of opinions and the transition matrix is $P = \widetilde{A}$, which is the row-normalized adjacency matrix.
The model \eqref{eq: dirac-model} has the same form as the DeGroot model of opinion dynamics \cite{degroot1974reaching}.

%%%%

\subsubsection{Exponential WTD}
We now suppose that each entity has an exponential WTD, which is closely related to a Poisson point process. In a Poisson point process, the time between two consecutive events follows an exponential distribution \cite{medhi1975waiting}. Poisson processes have been studied extensively both because they are mathematically convenient and because they are memoryless \cite{mahmud2016poisson, hallas1997waiting}. We denote the exponential WTD of entity $i$ by $T_{i}(t) = \lambda_ie^{-\lambda_it}$, where the rate parameter $\lambda_i > 0$. From a direct computation, we obtain
\begin{equation}
\begin{aligned}
    \theta_{i} (t)= \lambda_i\,, \quad \chi_i(t) = \lambda_i\delta(t)\,,
\end{aligned}
\end{equation}
which we insert into \eqref{eq: model-continuous} to obtain the opinion model
\begin{equation} \label{eq: exp-xi}
    \dot{x}_i(t) = \lambda_i \sum_j \widetilde{A}_{ij} x_j(t) - \lambda_ix_i(t)\,.
\end{equation}
This yields a Markovian dynamical process $x(t)$ that satisfies
\begin{equation} \label{eq: exp-model}
    x(t) = e^{\Lambda (\widetilde{A}-I)t}x(0)\,,
\end{equation}
where $\Lambda$ is a diagonal matrix with entries $\Lambda_{ii}=\lambda_i$, the matrix $\widetilde{A}$ is the row-normalized adjacency matrix, and $I$ is the identity matrix. If we discretize the continuous-time opinion $x(t)$ at times $n = 0, \, \triangle t, \, 2\triangle t, \ldots$, we have a discrete-time description for $x(n)$ that satisfies the iterative relation
\begin{equation}
    x(n+1) = e^{\Lambda (\widetilde{A}-I)\triangle t}x(n)\,.
\end{equation}
This discrete model is equivalent to the DeGroot model with a transition matrix $P=e^{\Lambda (\widetilde{A} - I)\triangle t}$ (instead of $\widetilde{A}$). We can thus view the model \eqref{eq: exp-xi} as a continuous-time extension of the DeGroot model.

%%%%%

\subsubsection{Gamma WTD}
Another WTD with exponential decay is the gamma distribution, which has been used widely for modeling a variety of phenomena, including human response times \cite{iribarren2011branching}, earthquake interevent times \cite{touati2009origin}, and {delayed effects in pharmacodynamics} \cite{sun1998transit}. 
Suppose that entity $i$ follows the gamma WTD $T_{i}(t) = \lambda^2_{i}t e^{-\lambda_it}$. From a direct computation, we obtain
\begin{equation}
\begin{aligned} \label{gamma}
    & \theta_{i} (t) = \frac{\lambda_i}{2} \left(  1-e^{-2\lambda_it}\right), \quad \chi_i(t) = \lambda_i^2 e^{-2\lambda_it}\,.
\end{aligned}
\end{equation}
We insert \eqref{gamma} into \eqref{eq: model-continuous} to obtain the non-Markovian opinion model
\begin{equation} \label{eq: gamma-model}
\begin{aligned}
    \dot{x}_i(t) &= \sum_j \frac12 \widetilde{A}_{ij}\lambda_i\left(1-e^{-2\lambda_i t}\right)x_j(t) - \int_0^t \lambda_i^2 e^{-2\lambda_i (t-t')}x_i(t')\,dt'\,,
\end{aligned}
\end{equation}
which we rewrite in matrix form as
\begin{equation} \label{eq: memory-matrix}
    \dot{x}(t) = \mathcal{K}(t)  \left[ \widetilde{A}  x(t) - \tilde{x}\{t\}\right]\,,
\end{equation}
where $\tilde{x}_i\{t\}$ is a historically averaged opinion that is weighted by the exponential kernel $\kappa_i(t) = \lambda_i^2 e^{-2\lambda_i t}$, which weights recent opinions more heavily than older opinions, and $\mathcal{K}(t)$ is the diagonal matrix with entries $\mathcal{K}_{ij}(t)=\delta_{ij}\mathcal{K}_{i}(t)$ that normalizes the kernel $\kappa_i$. Specifically, $\mathcal{K}_i(t)$ and $\tilde{x}_i\{t\}$ are 
\begin{equation} \label{eq: Ki}
  \mathcal{K}_i(t) = \int_0^t \kappa_i(t')\,dt' \,, \quad \tilde{x}_i\{t\}=\frac{\int_0^t \kappa_i(t-t') x_i(t')\,dt'}{\mathcal{K}_i(t)}\,.
\end{equation}
We use the curly-bracket notation $\{t\}$ to indicate that $\tilde{x}_i\{t\}$ depends on entity $i$'s entire opinion trajectory $\{x_i(t')\}_{t'\le t}$.
In comparison to the memoryless model \eqref{eq: exp-xi} that is induced by the exponential WTD, the memory-dependent model \eqref{eq: memory-matrix} includes an exponential time-relaxation kernel $\mathcal{K}_i(t)=\frac{\lambda_i}{2}\left(1-e^{-2\lambda_it}\right)$, which approaches the constant $\lambda_i/2$ as $t \rightarrow \infty$. The model \eqref{eq: memory-matrix} also includes a damping term that drives the opinion of each entity to its historical mean; this promotes self-consistency of each entity's opinion.

\add{If we define the integral term in Equation \eqref{eq: gamma-model} as an auxiliary variable $y_i$, we obtain the Markovian system
\begin{equation} \label{eq: gamma-ext}
\begin{aligned}
    \dot{x}_i(t) &= \sum_j \frac12 \widetilde{A}_{ij}\lambda_i\left(1-e^{-2\lambda_i t}\right)x_j(t) - y_i(t) \,, \\
    \dot{y}_i(t) &= \lambda_i^2 x_i(t) - 2\lambda_i y_i(t)\,,
\end{aligned}
\end{equation}
where the variables in this extended system are $x_i$ and $y_i$ for all $i \in \{1, \ldots, N\}$.
}

In Section \ref{sec: numerics}, we study the non-Markovian opinion model \eqref{eq: memory-matrix} that is induced by the Gamma WTD. We are not aware of any existing opinion models that have the same form as Equation \eqref{eq: memory-matrix}. 

%%%%%
 
\subsubsection{Heavy-tailed WTDs}
Many real systems, such as e-mail communication \cite{iribarren2009impact} and the spread of infectious diseases \cite{vazquez2007impact}, have bursty properties, which cannot be captured well by {Poisson} temporal statistics. 
In such situations, the time intervals between isolated events deviate from an exponential distribution. Instead, they follow a heavy-tailed distribution. 

Suppose that each node follows a Pareto WTD. The Pareto distribution has been used to model online participation inequality, distributions of wealth, website visits, and a variety of other phenomena \cite{wojcik2019sizing}.
We write the Pareto distribution in the form $T_{i}(t) = \lambda_i(t+1)^{-\lambda_i-1}$ with $\lambda_i > 0$, where we have shifted the distribution so that its domain is $[0,\infty)$.
A direct computation yields
\begin{equation}\label{eq: pareto-eq}
    \hat{T}_{i}(s) = \lambda_i\int_0^\infty (t+1)^{-\lambda_i-1}e^{-ts}\,dt \,, \quad \phi_i(t) = (t+1)^{-\lambda_i}\,.
\end{equation}
The opinion model in Equation \eqref{eq: model-discrete} involves the inverse Laplace transform of $\hat{\theta}_i = (1-\hat{T}_i)^{-1}\hat{T}_{i}$. Unfortunately, there is not a convenient formula for the inverse Laplace transform $\hat{T}_i$ in \eqref{eq: pareto-eq}.

Now suppose that the WTD of each node is a log-normal distribution, which is also heavy-tailed and has the PDF
\begin{equation}
    T_{i}(t) = \frac{1}{\sqrt{2\pi}\sigma_i t} \exp\left[ -\frac{(\ln t-\mu_i)^2}{2\sigma_i^2}\right]\,,
\end{equation}
where $\mu_i$ and $\sigma_i^2$, respectively, are the mean and variance of the Gaussian distribution. A closed-form expression does not exist for the Laplace transform of a log-normal distribution \cite{asmussen2016laplace}. Accordingly, we are unable to obtain closed-form expressions for related terms, such as $\hat{T}_i$ and $\hat{\theta}_i$ in \eqref{eq: thetai}, and their inverse Laplace transforms.

To the best of our knowledge, most common heavy-tailed distributions do not possess an explicit form for the opinion models \eqref{eq: model-discrete}. 
Instead of aiming to determine analytical expressions for models that are induced by heavy-tailed WTDs, we seek feasible numerical approaches to simulate opinion models \eqref{eq: model-discrete} that are induced by heavy-tailed WTDs.
\add{Masuda and Rocha \cite{masuda2018gillespie} proposed a fast Gillespie algorithm to simulate non-Poisson renewal processes with heavy-tailed interevent-time distributions. 
Their algorithm simulates agent-based trajectories $X_i(t)$ as interacting sequences of discrete events, but it does not generate the expected opinion value $x_i(t)$.}
In the present paper, our approach is to approximate the continuous-time WTDs in \eqref{eq: model-continuous} with sums of Dirac delta distributions. 
This yields opinion models \eqref{eq: model-discrete} that are induced by a sum of Dirac delta distributions.

%%%%

\subsubsection{WTDs that are sums of Dirac delta distributions} \label{delta-sum}
When a WTD consists of a sum of Dirac delta distributions, the events take place at a set of discrete times. The WTD $T_i$ of node $i$ is
\begin{equation} \label{eq: Ti-eq}
    T_{i}(t) = \sum_{k=1}^{\infty} m_k^i \, \delta(t - k\triangle_i)\,,
\end{equation}
where $m_k^i$, with $\sum_{k=1}^\infty m_k^i = 1$, is the probability that an event occurs after entity $i$ waits for time $k\triangle_i$. A direct computation yields
\begin{equation*} \label{eq: multi-delta-psi}
\begin{aligned}
    \phi_i(t) &= 1 - \sum_{k=1}^{\lfloor {t/\triangle_i} \rfloor} m_k^i\,, \quad  \theta_{i} (t) = \sum_{k = 0
    }^\infty M_k^i \delta(t-k\triangle_i)\,, \\
    M_k^i &= \sum_{\alpha\in U_k} m^i_{\alpha_1}m^i_{\alpha_2}\cdots m^i_{\alpha_z}\,, \quad U_k = \{\alpha \in \mathbb{N}_+^z: ~ \|\alpha\|_1=k, ~ z\in \mathbb{N}_+ \}\,,
\end{aligned}
\end{equation*}
where $\|a\|_1=|a_1|+\cdots+|a_z|$ is the discrete $\ell_1$ norm, $\lfloor \cdot \rfloor$ is the floor function, and $\mathbb{N}_+ = \{ 1, 2, \ldots\}$. When $k = 0$, we define $U_0$ to be the empty set, which implies that $M_0^i=0$.
The first four terms of the sequence $M_k^i$ are
\begin{align*}
    M_0^i &= 0\,, \quad M_1^i = m_1^i \,, \\
    M_2^i &= m_1^im_1^i+m_2^i\,, \quad M_3^i = m_1^im_1^im_1^i+m_1^im_2^i+m_2^im_1^i+m_3^i\,.
\end{align*}
When the node WTDs are sums of Dirac delta distributions, the associated memory-dependent
opinion model is
\begin{equation} \label{eq: multi-delta-x}
    x_i(t) = \sum_j \widetilde{A}_{ij} \sum_{k=1}^{\lfloor{t/\triangle_i}\rfloor}\phi_i(t-k\triangle_i)M_k^ix_j^-(k\triangle_i) + \phi_i(t)x_i(0)\,.
\end{equation}
This model accounts for situations in which each entity updates its opinions at discrete times. 
One can use models of the form \eqref{eq: multi-delta-x} as approximate models that are induced by heavy-tailed WTDs.
When the inverse Laplace transforms of the node WTDs are difficult to compute, we can discretize the WTDs (see Equation \eqref{eq: Ti-eq}) with a small time step $\triangle_i$ and treat \eqref{eq: multi-delta-x} as the resulting model of opinion dynamics.

When the time steps are uniform (i.e., $\triangle_i = \triangle t$ for each $i$), we write \eqref{eq: multi-delta-x} in matrix form
\begin{equation} \label{eq: discrete-fj}
    x[n+1] = \sum_{k=0}^n \Lambda_{\phi}[n-k]\Lambda_{M}[k+1]\widetilde{A}x[k] + \Lambda_{\phi}[n+1]x[0]\,,
\end{equation}
where the bracket $[k]$ denotes evaluation at time $k\triangle t$ (e.g., $x[n]=x(n\triangle t)$) and $\Lambda_\phi$ and $\Lambda_M$ are diagonal matrices with entries
\begin{equation} \label{eq: Lambda_phi_M}
     \left(\Lambda_{\phi}[k]\right)_{ii} = \phi_i(k)\,, \quad \left(\Lambda_{M}[k]\right)_{ii} = M_k^i\,, \quad k \in\mathbb{N}\,.
\end{equation}
In the classical Friedkin--Johnsen (FJ) model \cite{friedkin2011social}, the opinion updates follow the rule $x[n+1] = \eta Px[n] + (1-\eta)x[0]$ with a time-independent weight $(1 - \eta)$ on the nodes' initial opinions. 
We can view the model \eqref{eq: discrete-fj} as an extension of the FJ model by writing $x[n+1]$ as a sum of all previous opinions $x[k]$ (with $k \in \{0, \ldots, n\}$) weighted by time-dependent matrices $P(n+1,k)$. That is,
\begin{equation} \label{eq: multi-delta-homo}
    x[n+1] = \sum_{k=0}^n P(n+1,k)x[k]\,.
\end{equation}
Each time step generates a new opinion vector $x[n]$, which joins the collection $\{x[0], x[1], \ldots, x[n]\}$ of historical opinions that collectively determine $x[n+1]$. The model \eqref{eq: multi-delta-homo} renormalizes the weights of the historical opinions $x[k]$ at each time step and introduces a temporal dependence on the weight matrices $P(n+1,k)$. 
The model \eqref{eq: discrete-fj} is also related to a voter model with an exogenous updating rule \cite{fernandez2011update}, {and this voter model is a special case of \eqref{eq: discrete-fj}}. 
In Section \ref{sec: numerics}, we implement the model \eqref{eq: multi-delta-x} with weights $M_k^i$ from continuous-time distributions.

%%%%%

\subsection{Theoretical analysis} \label{sec: analysis}
We now discuss the properties --- including opinion conservation, convergence, and conditions for consensus --- of the proposed memory-dependent opinion models \eqref{eq: model-discrete} for both homogeneous and heterogeneous scenarios.

%%%%

\subsubsection{Conservation of a weighted average of the opinions in homogeneous models}

When all nodes have the same WTD, the opinion models \eqref{eq: model-discrete} reduce to the homogeneous models \eqref{eq: homo}. 
We first state an opinion-conservation guarantee for the homogeneous models \eqref{eq: homo}.

\begin{theorem} \label{thm: mean}
\add{Let $w$ be a left eigenvector of $\widetilde{A}$ with eigenvalue $1$ (i.e., $w^TA = w^T$). 
Suppose that $x_i(t)$ are solutions of the homogeneous opinion models \eqref{eq: homo}. It then follows that the averaged opinion 
\begin{equation} \label{eq: weight}
    \bar{x}(t) = \sum_i w_i x_i(t)
\end{equation}
is conserved for any WTD.}
\end{theorem}

\begin{proof}
\add{Because $\widetilde{A}$ is row-normalized, it has the eigenvalue $1$ and an associated real left eigenvector $w$. Using Equations \eqref{eq: homo} and \eqref{eq: weight}, we obtain the scalar integral equation}
\begin{equation} \label{eq: barx}
    \bar{x}(t) = \int_0^t \phi(t-t')\psi(t')\bar{x}^-(t')\,dt'\,,
\end{equation}
where $\bar{x}^-(t) = \lim_{t'\rightarrow t_-}\bar{x}(t')$ and $\phi$ and $\psi$ are defined in Equations \eqref{eq: phi-eq} and \eqref{eq: psi-prop}, respectively.
Let $t_\text{max}$ and $t_\text{min}$ be the times that $\bar{x}(t)$ takes its maximum value and minimum value, respectively, in the time interval $[0,T]$. That is,
\begin{equation}
    t_\text{max} = \text{arg max}_{t \in [0,T]} ~\bar{x}(t)\,, \quad t_\text{min} = \text{arg min}_{t \in [0,T]} ~\bar{x}(t)\,.
\end{equation}
Recall that $\phi$ and $\psi$ are positive and that $\phi\star\psi=1$. Direct computations yield
\begin{subequations}\label{bounds}
\begin{align} 
    \bar{x}(t_\text{max}) &\le \bar{x}(t_\text{max})\int_0^{t_\text{max}}  \phi(t_\text{max}-t')\psi(t')\,dt' = \bar{x}(t_\text{max})\,, \label{upper-bounds}\\
    \bar{x}(t_\text{min})  &\ge \bar{x}(t_\text{min})\int_0^{t_\text{min}} \phi(t_\text{min}-t')\psi(t')\,dt' = \bar{x}(t_\text{min})\,. \label{lower-bounds}
\end{align}
\end{subequations}
The inequality in \eqref{upper-bounds} is an equality only if $\bar{x}(t) = \bar{x}(t_\text{max})$ for all $t \in [0,t_\text{max}]$, and the inequality in \eqref{lower-bounds} is an equality only if $\bar{x}(t) = \bar{x}(t_\text{min})$ for all $t \in [0,t_\text{min}]$.
This implies that $\bar{x}(0)=\bar{x}(t_\text{max})=\bar{x}(t_\text{min})$ and hence that $\bar{x}(t)$ is constant on any finite time interval $[0,T]$. The choice of $T$ is arbitrary, so $\bar{x}(t)$ is constant.
\end{proof}

\add{Researchers have previously noted the conservation of weighted averages in models of opinion dynamics~\cite{sood2008voter,masuda2009evolutionary}.
When the row-normalized adjacency matrix $\widetilde{A}$ is also column-normalized, Theorem \ref{thm: mean} implies that the mean opinion $\bar{x}(t) = \sum_i x_i(t)/N$ is conserved (i.e., $\bar{x}(t)$ is constant).}
For an arbitrary heterogeneous model (which can take versatile forms), it is not guaranteed that the mean opinion is conserved. In Figure \ref{fig: x-active-stubborn}, we show an example of a heterogeneous model in which the mean opinion moves towards the opinions of nodes that have larger expected waiting times. In Section \ref{sec: numerics}, we discuss this example in detail.

%%%%%

\subsubsection{Analysis of consensus for homogeneous models}

\add{One common question in models of opinion dynamics is whether or not a model  
converges to a consensus, in which all entities hold the same opinion.} Researchers have successfully determined consensus conditions and the time to reach consensus in several types of models, including the classical DeGroot model \cite{degroot1974reaching}, the FJ model \cite{friedkin1990social}, and BCMs on graphs \cite{lorenz2005stabilization} and hypergraphs \cite{hickok2022bounded,chu2022density}. We analyze the steady-state opinions in the homogeneous memory-dependent models \eqref{eq: homo} and provide sufficient conditions that guarantee convergence to consensus. 

\begin{theorem} \label{thm: consensus}
If the row-normalized adjacency matrix $\widetilde{A}$ is diagonalizable and $-1$ is not an eigenvalue of $\widetilde{A}$, then the homogeneous models \eqref{eq: homo} that are induced by any WTD converge to the same steady state $x^*$. The steady state $x^*$ satisfies
\begin{equation}
    \lim_{t\rightarrow \infty} x(t) = x^* = \!\! \sum_{\{d:~\nu_d=1\}} c_d^0v_d\,,
\end{equation}
where $\{v_d\}$ are the eigenvectors (which we assume are linearly independent) of $\widetilde{A}$ and $\{c_d^0\}$ are the associated coefficients of $x(0)$ in the basis $\{v_d\}$. 
\end{theorem}
\begin{proof}
Let $\{\nu_d\}$ be the eigenvalues of $\widetilde{A}$. Consider the decomposition $x(t)=\sum_d c_d(t)v_d$. From Equation \eqref{eq: homo}, we know that the basis coefficients $c_d(t)$ satisfy
\begin{equation} \label{eq: ci-eq}
    c_d(t) = \nu_d \left[\phi\star\left( \psi c_d^- \right)\right](t) + (1-\nu_d)\phi(t)c_d^0\,, 
\end{equation}
where $\phi$ and $\psi$ are defined in \eqref{eq: phi-eq} and \eqref{eq: psi-prop}, respectively. Because $\widetilde{A}$ is a right stochastic matrix (i.e., its row sums are $1$), the eigenvalues satisfy $|\nu_d|\le1$. 
By assumption, $\nu_d\neq -1$. We discuss the two cases $\nu_d = 1$ and $|\nu_d| < 1$ separately.

For eigenvalues $\nu_d = 1$, we rewrite Equation \eqref{eq: ci-eq} as
\begin{equation}
    c_d = \phi \star \left( \psi c_d^- \right)\,, 
\end{equation}
where the coefficients $c_d$ satisfy the equation for $\bar{x}$ in Equation \eqref{eq: barx}. By Theorem \ref{thm: mean}, we know that $c_d(t)$ remains constant and hence always equals its initial value $c_d^0$. 

For eigenvalues $|\nu_d|<1$, we consider a mapping $F_d$ between {the $L^\infty$-function space that is} equipped with the $L^\infty$ norm $\|y\|_{L^\infty}=\sup_{t\in[0,\infty)} |y(t)|$. For all $y\in L^\infty$, we define the mapping $F_d$ with the equation
\begin{equation}
    F_d[y](t) = \nu_d \left[\phi\star\left( \psi y^- \right)\right](t) + (1-\nu_d)\phi(t)c_d^0\,.
\end{equation}
Because $\phi\star\psi=1$, we have
\begin{equation} \label{eq: Fic}
\begin{aligned}
    \|F_d(y_1-y_2)\|_{L^\infty} &=  \left\|\nu_d \phi\star\left[\psi \left(y_1^--y_2^-\right) \right]\right\|_{L^\infty} \\
    	&\le  |\nu_d| \|y_1-y_2\|_{L^\infty} \| \phi \star \psi\|_{L^\infty} = |\nu_d| \|y_1-y_2\|_{L^\infty}\,,
\end{aligned}
\end{equation}
which implies that $F_d$ is a contraction mapping. According to the Banach fixed-point theorem, there exists a unique fixed point $c_*\in L^\infty$ that satisfies $F_d[c_*] = c_*$. Therefore, $c_*$ is the unique solution of Equation \eqref{eq: ci-eq}. For any fixed $t$, we choose $T > t$ and have
\begin{equation} \label{eq: c*}
    |c_*(T)| \le \|\nu_dc_*\|_{L^{\infty}}\int_0^t \phi(T-t')\psi(t')\,dt' + |\nu_d|\sup_{t'\ge t}|c_*(t')| + |(1-\nu_d)c_d^0|\phi(T)\,.
\end{equation}
We let $T \rightarrow \infty$ in Equation \eqref{eq: c*} to obtain
\begin{equation} \label{eq: supc}
    \overline{\lim}_{T\rightarrow \infty} |c_*(T)|\le |\nu_d|\sup_{t'\ge t}|c_*(t')|\,,
\end{equation}
where we have used the facts that $\phi(t)$ is nonincreasing and $\lim_{t\rightarrow \infty}\phi(t) = 0$. The time $t$ is arbitrary, so we let $t \rightarrow \infty$ in Equation \eqref{eq: supc} and obtain 
\begin{equation} \label{245}
    \overline{\lim}_{T\rightarrow \infty} |c_*(T)|\le |\nu_d|\overline{\lim}_{t\rightarrow \infty} |c_*(t)|\,.
\end{equation}
Because $|\nu_d|<1$, the inequality \eqref{245} holds if and only if $\text{lim}_{t\rightarrow \infty} |c_*(t)|=0$. Therefore, for all $|\nu_d|<1$, there is a unique solution $c_d(t)$ of Equation \eqref{eq: ci-eq} with $c_d(0) = c_d^0$ and $\lim_{t\rightarrow\infty}c_d(t) = 0$.

Combining the cases for $\nu_d = 1$ and $|\nu_d| < 1$, we have
\begin{equation}
    \lim_{t\rightarrow \infty} x(t) = \lim_{t\rightarrow \infty} \left[ \sum_{\{d:~|\nu_d|<1\}} \!\! c_d(t)v_d+\!\! \sum_{\{d:~\nu_d=1\}} \!\! c_d(t)v_d \right] = \sum_{\{d:~\nu_d=1\}} \!\! c_d^0v_d \,.
\end{equation}
\end{proof}

\medskip

\begin{remark} In Theorem \ref{thm: consensus}, it is necessary to include the condition that $\widetilde{A}$ does not have an eigenvalue of $-1$. Consider the $2 \times 2$ matrix 
\begin{equation} \label{example}
	\widetilde{A} = A =  \begin{pmatrix}
		0 & 1 \\ 1 & 0
		\end{pmatrix}\,,
\end{equation}
whose eigenvalues are $1$ and $-1$. The model \eqref{eq: dirac-model} that is induced by the Dirac delta WTD never converges to a steady state if the two opinions are different initially. The two entities swap their opinions whenever an event occurs.
\end{remark}

Using Theorem \ref{thm: consensus}, we find the following sufficient conditions to guarantee that the homogeneous opinion models \eqref{eq: homo} converge to consensus 
\add{(i.e., $\lim_{t\rightarrow \infty} x_i(t)=x_\text{same}$ for all $i$).}

\begin{corollary} \label{thm: irreducible-homo}
If the row-normalized adjacency matrix $\widetilde{A}$ is irreducible and does not have an eigenvalue of $-1$, then all homogeneous models \eqref{eq: homo} converge to consensus regardless of the initial conditions and WTDs.
\end{corollary}

\begin{corollary} \label{thm: initial_coro}
Suppose that $\widetilde{A}$ is diagonalizable and does not have an eigenvalue of $-1$. If the decomposition of $x(0)$ satisfies $\sum_{\{d:~\nu_d=1\}} c_d^0v_d = c\mathbb{1}$, where $c$ is a scalar and $\mathbb{1}$ is a vector in which each entry is $1$, then the homogeneous models \eqref{eq: homo} converge to consensus with the opinion value $c$ for any WTD.
\end{corollary}
\begin{corollary} \label{thm: degroot-convergence}
If the homogeneous DeGroot model \eqref{eq: dirac-model} (which is induced by the Dirac delta WTD) converges to consensus, then all homogeneous models \eqref{eq: homo} converge to consensus.
\end{corollary}

These three corollaries are direct consequences of Theorem \ref{thm: consensus}. 
In Corollary \ref{thm: degroot-convergence}, the convergence to consensus of any particular homogeneous model \eqref{eq: homo} other than the DeGroot model does not imply that the DeGroot model also converges to consensus. For example, consider the adjacency matrix in \eqref{example}. The model \eqref{eq: exp-model} that is induced by the exponential WTD converges to consensus for any initial state, but the DeGroot model never converges if the initial opinions are different. In numerical computations, we observe for the adjacency matrix in \eqref{example} that the opinion models that are induced by the uniform WTD, the gamma WTD, and heavy-tailed WTDs also converge to consensus. For continuous WTDs, the events of two entities occur simultaneously with probability $0$. With probability $1$, the event of one entity occurs first, which causes an opinion adoption by the other entity and ultimately leads to consensus.

%%%%%

\subsubsection{Analysis of consensus for heterogeneous models with Poisson statistics}

When the WTDs are exponential, the interevent times arise from Poisson point processes. In the following theorem, we state a convergence condition for this situation.

\begin{theorem} \label{thm: consensus-heter}
Suppose that all nodes have exponential WTDs, which we parameterize by the rate parameter $\lambda_i > 0$ for node $i$. 
Let $\Lambda$ be the diagonal matrix with entries $\Lambda_{ii} = \lambda_i$, and let $\widetilde{A}$ be a row-normalized adjacency matrix. If the matrix $Z = \Lambda(\widetilde{A} - I)$ is diagonalizable with eigenvalue--eigenvector pairs $\{\nu_d, v_d\}_{d=1,\ldots,N}$ and $\{v_d\}$ are linearly independent, then the model \eqref{eq: exp-xi} converges to a steady state $x^*$. Additionally, $x^*$ satisfies
\begin{equation}\label{eq: steady-state-heter}
    \lim_{t\rightarrow \infty} x(t) = x^* = \!\! \sum_{\{d:~\nu_d=0\}} c_d^0v_d\,,
\end{equation}
where $c_d^0$ is the coefficient of $v_d$ in the decomposition of $x(0)$ in terms of the basis $\{v_d\}_{d=1,\ldots,N}$. 
\end{theorem}

\begin{proof}
The solution of Equation \eqref{eq: exp-xi} satisfies $x(t) = e^{Zt}x(0)$. If we express $x(t)$ using the basis $\{v_d\}$, then the coefficients $c_d(t)$ satisfy $c_d(t) = e^{\nu_dt}c_d^0$. Let $\text{eig}_\text{max}(\mathcal{M})$ denote the maximum eigenvalue of a matrix $\mathcal{M}$.
For all eigenvalues $\nu_d$, we have
\begin{equation}
    \nu_d \le \text{eig}_\text{max}(Z) = \text{eig}_\text{max}\left(\Lambda(\widetilde{A}-I)\right) \le \text{eig}_\text{max}(\Lambda) \, \text{eig}_\text{max}(\widetilde{A}-I) \,.
\end{equation} 
By the Gershgorin circle theorem, the maximum eigenvalue of $\widetilde{A} - I$ is less than or equal to $0$, which implies that $\nu_d \le 0$. Consequently, the coefficients $c_d(t) = e^{\nu_dt}c_d^0$ satisfy $\lim_{t\rightarrow \infty} c_d(t) = 0$ for $\nu_d < 0$ and $c_d(t) = c_d^0$ for $\nu_d = 0$. This concludes the proof. 
\end{proof}

Theorem \ref{thm: consensus-heter} gives a convergence condition for a heterogeneous model \eqref{eq: exp-xi} that is induced by exponential WTDs. The matrix $Z$ and the initial condition together determine if a model converges to a consensus state. However, the rate parameter $\lambda_i$ of the exponential WTD affects the speed of convergence; a larger $\lambda_i$ results in faster convergence. Using the same notation as in Theorem \ref{thm: consensus-heter}, the following corollary guarantees convergence to consensus.

\begin{corollary} \label{this-cor}
If $Z$ is diagonalizable and irreducible, then the model \eqref{eq: exp-xi} converges to consensus.
\end{corollary}

\add{Corollary \ref{this-cor} is a direct consequence of Theorem \ref{thm: consensus-heter}. Because $\widetilde{A}$ is row-normalized, $Z = \Lambda(\widetilde{A} - I)$ has the eigenvalue $0$ with the associated eigenvector $\mathbb{1}$. 
The irreduciblity condition in Corollary \ref{this-cor} guarantees that $\mathbb{1}$ is the only eigenvector of the eigenvalue $0$ and guarantees convergence to consensus.
}

There are other sufficient conditions that guarantee convergence to consensus. For example, requiring the initial opinion to have a decomposition of the form in Corollary \ref{thm: initial_coro} also guarantees convergence to consensus for the heterogeneous model \eqref{eq: exp-xi}.

%%%%%%

%
\section{Numerical computations} \label{sec: numerics}
In this section, we numerically investigate the time evolution of the opinion models that we proposed in Section \ref{sec: models}. We study how WTDs affect opinions dynamics by examining steady-state opinion clusters and the time to converge to a steady state for both homogeneous and heterogeneous models on a variety of graphs. 

We compare the opinion models \eqref{eq: model-discrete} with different WTDs on three types of graphs. The first graph is the largest connected component of the Caltech network from the {\sc Facebook100} data set \cite{traud2012social,red2011comparing}. The nodes are individuals and the edges encode Facebook ``friendships'' between those individuals on one day in fall 2005. 
The second type of network is a graph that we generate using the stochastic block model (SBM) $G(N,s,p,q)$, where $N$ denotes the number of nodes. We assign nodes uniformly at random to one of $s$ communities. We place edges between nodes in the same community with homogeneous and independent probability $p$, and we place edges between nodes in different communities with homogeneous and independent probability $q$. 
We run our simulations on only one SBM graph, but we expect to obtain similar results on other graphs that are generated by the same SBM.
The third type of network is a complete weighted graph with random edge weights, which we draw independently from the uniform distribution on $[0,1]$. 
We generate 100 graphs from this network ensemble and use them in the study of the heterogeneous models \eqref{eq: model-discrete}. 
In Figure \ref{fig: graphs}, we show the sparsity patterns of the adjacency matrices of these graphs.

\begin{figure}[htp]
    \centering
    \includegraphics[width=0.9\textwidth]{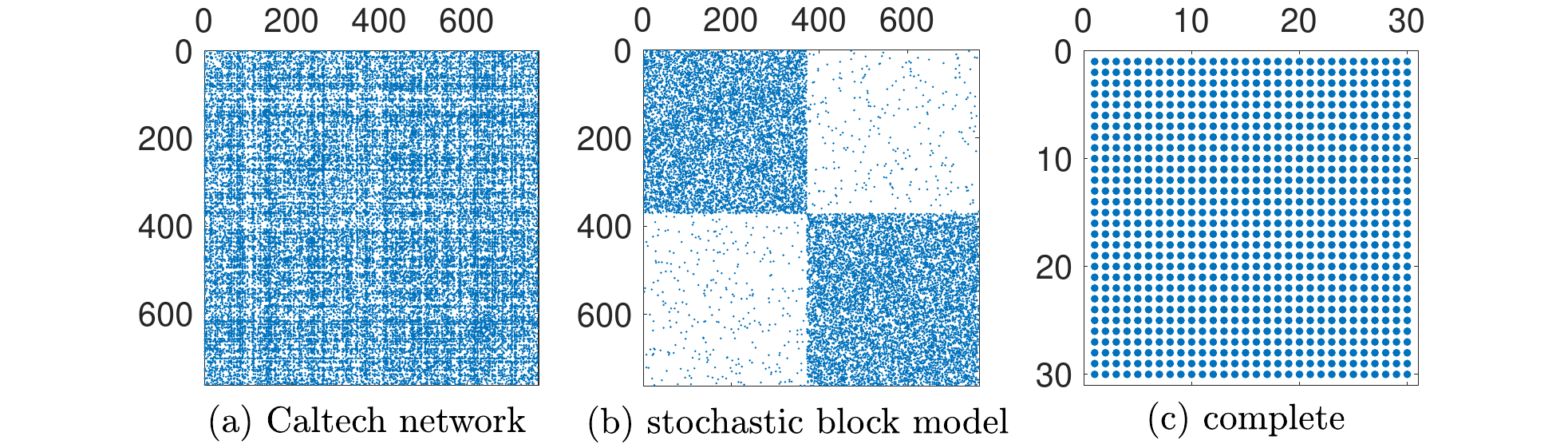}
    \caption{Sparsity patterns of the adjacency matrices of three graphs. Each blue dot signifies a nonzero entry and each white dot signifies a $0$ entry. (a) The Caltech network has $N = 762$ nodes and 16,651 edges. (b) In our stochastic-block-model network, there are $N = 762$ nodes, $s = 2$ communities, an edge probability of $p \approx 0.0554$ within communities, and an edge probability of $q = 0.002$ between communities. With these probabilities, the expected total number of edges matches the number of edges in the Caltech network. (c) A complete weighted graph with $N = 30$ nodes and weights that we draw independently and uniformly from the interval $[0,1]$. 
    }
    \label{fig: graphs}
\end{figure}

\add{We consider six different WTDs --- the Dirac delta, exponential, gamma, uniform, Pareto, and log-normal distributions. The WTDs are given by the formulas
\begin{subequations} \label{eq: Ts}
\begin{align} 
    T_\text{delta}(t) &= \delta(t - \mu)\,, \label{eq: Tdelta}\\
    T_\text{uniform}(t) &= \mathbb{1}_{[0,2\mu]}(t)\,,  \label{eq: Tuniform}\\
    T_\text{gamma}(t) &= \frac{4t}{\mu^2}\exp(-2t/\mu)\,, \label{eq: Tgamma}\\
     T_\text{exp}(t) &= \frac{1}{\mu}\exp(-t/\mu) \,, \label{eq: Texp}\\
    T_\text{LN}(t) &= \frac{1}{\sqrt{2\pi}t\sigma}\exp\left(-\frac{\left(\ln(t)-\mu\right)^2}{2\sigma^2}\right) \,, \label{eq: TLN} \\ 
    T_\text{pareto}(t) &= \frac{\alpha}{(1 + t)^{\alpha + 1}}  \,. \label{eq: Tpareto}
\end{align}
\end{subequations}
The mean of the WTDs in Equations \eqref{eq: Tdelta}--\eqref{eq: Texp} is equal to $\mu$. 
The mean of the log-normal WTD in \eqref{eq: TLN} is $\exp(\mu+\frac{\sigma^2}{2})$, where the parameters $\mu$ and $\sigma$ are the mean and  standard deviation of a normal distribution. 
The mean of the Pareto WTD in \eqref{eq: Tpareto} is $1/(\alpha-1)$, where the parameter $\alpha > 1$ is called the Pareto index.}
The Dirac delta and uniform distributions are compactly supported, the gamma and exponential distributions have light tails (specifically, they decay exponentially), and the log-normal and Pareto distributions have heavy tails.
We discretize the above WTDs (except for the Dirac delta distribution, which is already discrete) using uniform grids (with a spacing of $0.01$ between grid points) and approximate them by a sum of Dirac delta distributions (see Equation \eqref{eq: Ti-eq}). 

%%%%%

\subsection{Models of opinion dynamics with homogeneous WTDs}
\label{sec: numeric-homo}

In this subsection, we investigate the time evolution of opinion models \eqref{eq: homo} with homogeneous WTDs on the Caltech network and the SBM network. 
These two networks have the same number of nodes and a similar number of edges, but they have different community structures. 
In the WTDs in \eqref{eq: Ts}, we use the following parameter values: $\mu = 1$ in \eqref{eq: Tdelta}--\eqref{eq: Texp}, $\mu = -1$ and $\sigma = \sqrt{2}$ in \eqref{eq: TLN}, and $\alpha = 2$ in \eqref{eq: Tpareto}\,. With these choices, these WTDs all have the same mean, which is equal to $1$.

We first examine the homogeneous opinion models \eqref{eq: homo} on the Caltech network. Because the row-normalized adjacency matrix $\widetilde{A}$ is irreducible and does not have an eigenvalue of $-1$, by Corollary \ref{thm: irreducible-homo}, we expect the opinions to converge to consensus. In Figure \ref{fig: x-caltech}, we show the opinion trajectory $x_i(t)$ of each entity $i$ for each WTD. We also show the time-dependent basis coefficients $c_d(t)$ that we obtain by expressing the time-dependent opinion vector $x(t)$ in terms of the eigenvectors $v_d$ of $\widetilde{A}$.

\begin{figure}[htbp]
    \centering
    \includegraphics[height=0.45\textwidth]{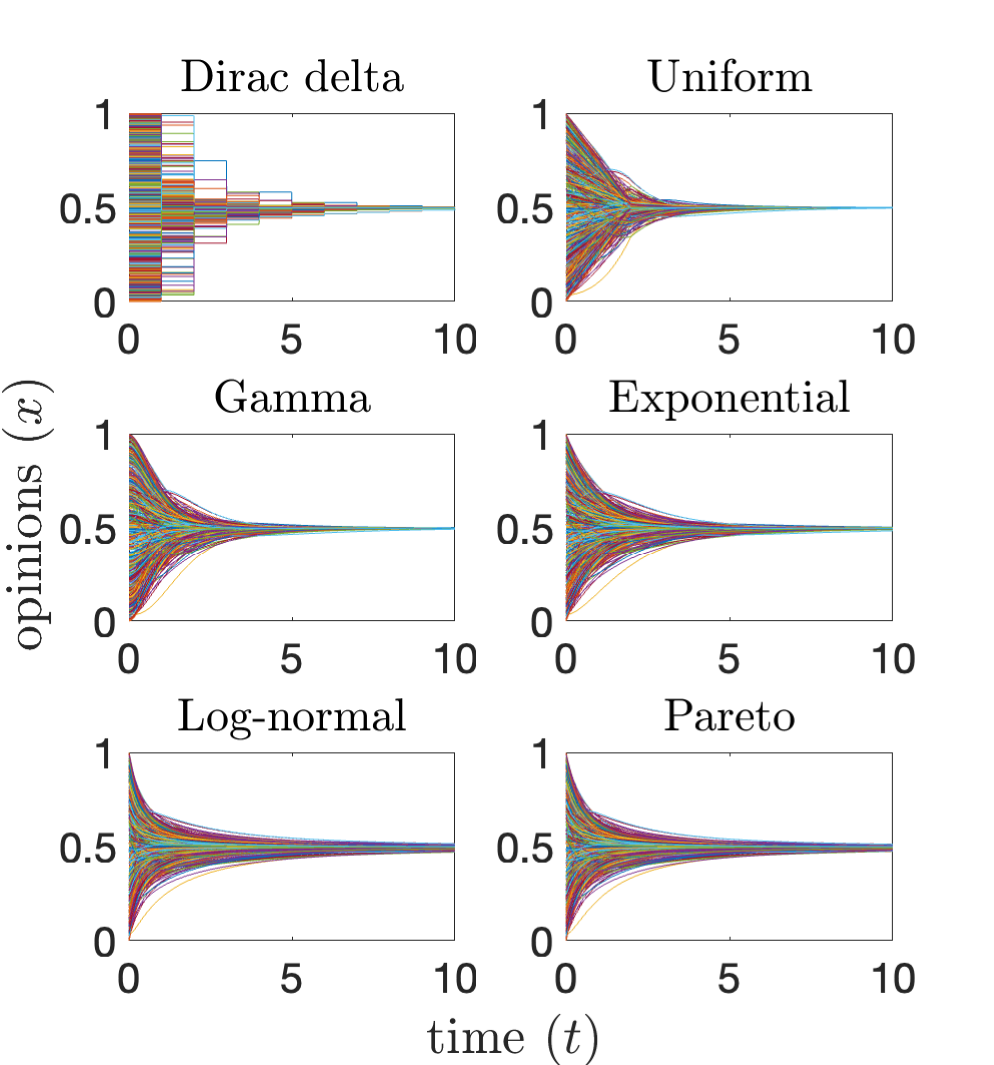}
    \includegraphics[height=0.45\textwidth]{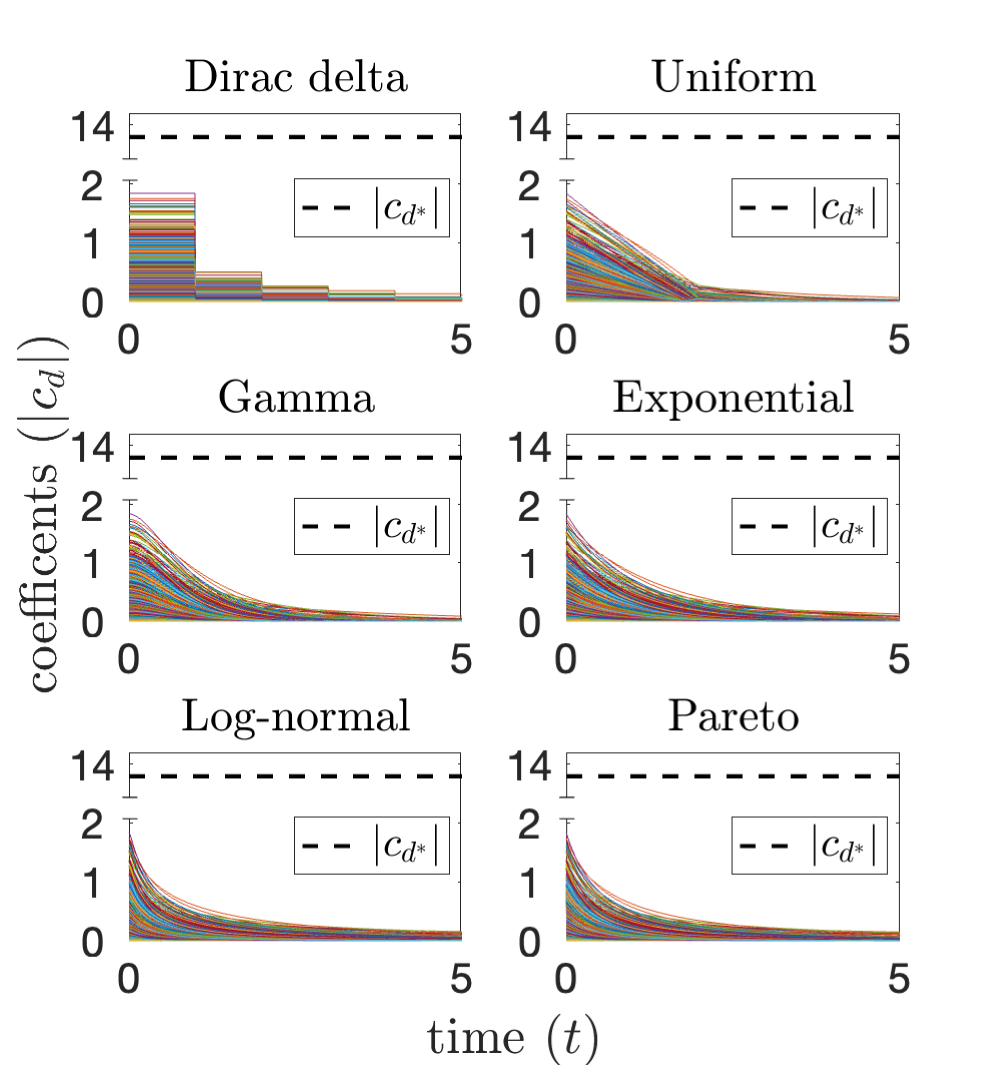}
  \caption{(Color online) Opinion trajectories $x_i(t)$ and their associated basis coefficients $c_d(t)$ for the homogeneous opinion models \eqref{eq: homo} with different WTDs on the Caltech network. All of the models have the same initial opinion, which we draw randomly from the uniform distribution on $[0,1]$. In the right panels, we plot the magnitudes $|c_d(t)|$ of the basis coefficients as a function of time. We use solid colored curves for the coefficients that are associated with eigenvalues that are smaller than $1$, and we use the dashed black lines (which we label with $|c_{d^*}|$) to plot the coefficients that are associated with the leading eigenvalue $\nu_{d^*} = 1$. We observe that $|c_{d^*}(t)| \approx 13.79$ for each of the WTDs.}
    \label{fig: x-caltech}
\end{figure}

As expected, the opinions of all entities converge to a single opinion cluster for each type of WTD. We also observe that the coefficient $c_{d^*}(t)$ that is associated with the eigenvalue $\nu_{d^*} = 1$ is constant with respect to time and that the magnitudes $|c_d(t)|$ of the other coefficients decay to $0$ for all WTDs. For different WTDs, the coefficient magnitudes $|c_d(t)|$ have different dynamics as they decay to $0$. For example, the coefficient magnitudes $|c_d(t)|$ for the uniform WTD $T_\text{uniform}(t) = \mathbb{1}_{[0,2]}(t)$ decay linearly at first. 
Recall the decomposition in \eqref{eq: ci-eq} in Theorem \ref{thm: consensus}. The basis coefficients $c_d(t)$ satisfy the bound
\begin{equation} 
    |c_d(t)| \le |\nu_d| \|c_d\|_{L^\infty} + \phi(t)|(1-\nu_d)c_d^0|\,.
\end{equation}
When $T$ is a uniform PDF, we compute from \eqref{eq: phi-eq} that $\phi(t)=1 - t/2$ for $t\in[0,2]$ and $\phi(t)=0$ otherwise. This formula for $\phi$ explains the associated linear trend for the coefficients $c_d(t)$ in Figure \ref{fig: x-caltech}. 

We generate a single two-community SBM graph (see Figure \ref{fig: graphs}(b)) and examine the homogeneous opinion models \eqref{eq: homo} on that SBM graph. We compute the eigenvalues of the row-normalized adjacency matrix $\widetilde{A}$. Despite the two-community structure of this network, $\widetilde{A}$ has only one eigenvalue that is equal to $1$. As guaranteed by Theorem \ref{thm: consensus}, the opinions converge to consensus for all initial opinions and all WTDs. In Figure \ref{fig: x-sbm}, we show the opinion trajectories $x_i(t)$ and their associated coefficients $c_d(t)$. For each WTD, the opinions converge to consensus and the magnitudes $|c_d(t)|$ of the coefficients decay to $0$ when $\nu_d \neq 1$.

\begin{figure}[htbp]
    \centering
    \includegraphics[height=0.45\textwidth]{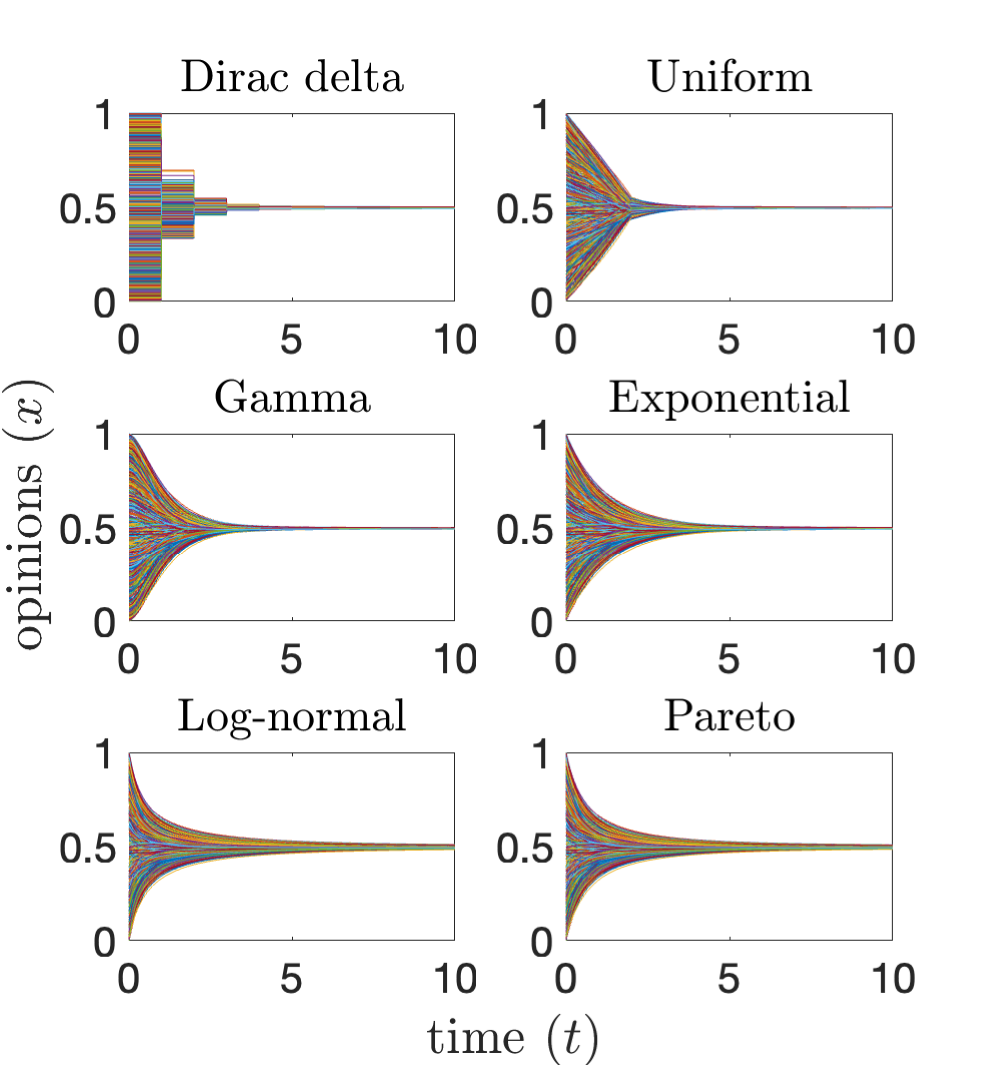}
    \includegraphics[height=0.45\textwidth]{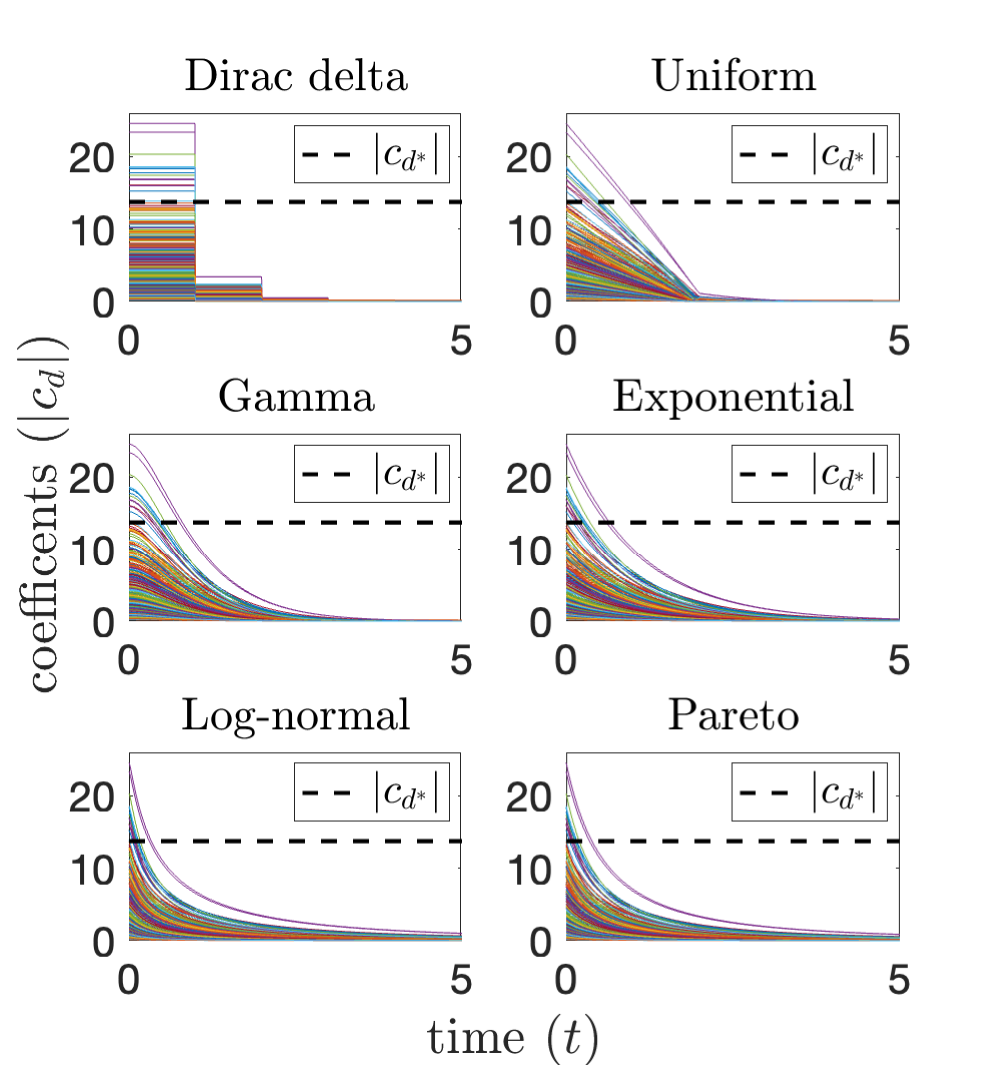}
    \caption{(Color online) Opinion trajectories $x_i(t)$ and their associated basis coefficients $c_d(t)$ for the homogeneous opinion models \eqref{eq: homo} with different WTDs on the same two-community SBM network. All of the models have the same initial opinion, which we draw randomly from the uniform distribution on $[0,1]$. 
    In the right panels, we plot the magnitudes $|c_d(t)|$ of the basis coefficients as a function of time.
    We use solid colored curves for the coefficients that are associated with eigenvalues that are smaller than $1$, and we use the black dashed lines (which we label with $|c_{d^*}|$) to plot the coefficients that are associated with the leading eigenvalue $\nu_{d^*} = 1$. 
    We observe that $|c_{d^*}(t)| \approx 13.42$ for each of the WTDs.
    }
    \label{fig: x-sbm}
\end{figure}

In the above simulations, we observe that the homogeneous opinion models \eqref{eq: homo} that are induced by any WTD converge to consensus for both the Caltech and SBM networks, but different models have different convergence rates. 
To compare the convergence rates of the homogeneous opinion models \eqref{eq: homo} that are induced by our different WTDs, we plot time-dependent variances of the opinions in Figure \ref{fig: convergence}. 
For both networks, the variance decays exponentially with time and the opinion models that are induced by the heavy-tailed WTDs (i.e., the log-normal and Pareto distributions) converge the slowest to steady state; the models that are induced by the uniform and gamma distributions converge the fastest.

\begin{figure}[htbp] 
    \centering
    \includegraphics[height=0.25\textwidth]{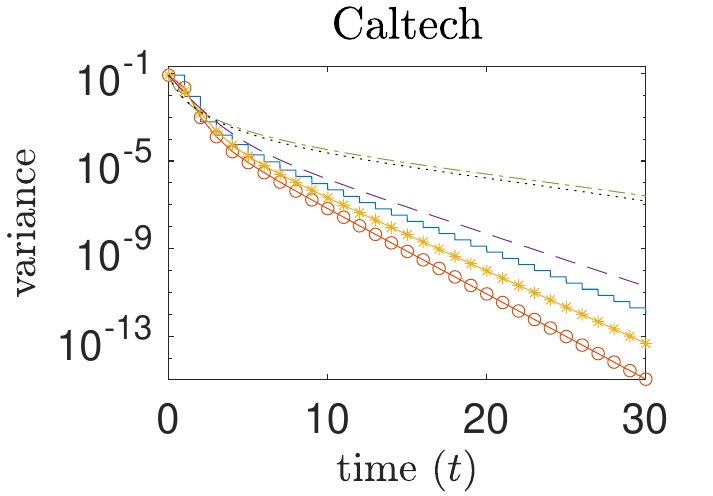}
    \includegraphics[height=0.25\textwidth]{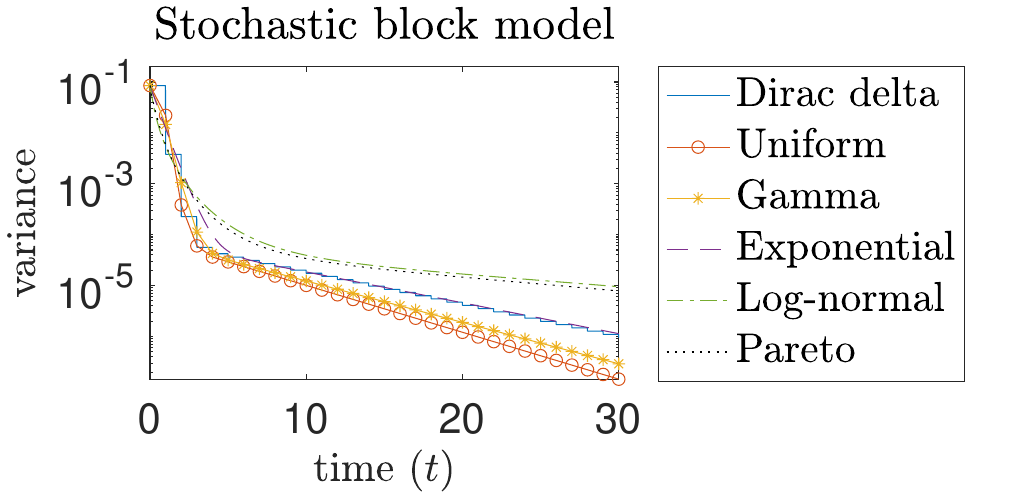}
    \caption{Time-dependent variances of the node opinions for opinion models with homogeneous WTDs on (left) the Caltech network and (right) a two-community SBM network.
    } 
    \label{fig: convergence}
\end{figure}

In both networks, the row-normalized adjacency matrices $\widetilde{A}$ have the same largest eigenvalue of $1$. Because the coefficient $c_{d^*}$ that is associated with the eigenvalue $1$ is constant as a function of time, the second-largest eigenvalue of $\widetilde{A}$ determines the convergence rate. 
In the Caltech network, the second-largest eigenvalue of $\widetilde{A}$ is $0.7229$. In the SBM network, the second-largest eigenvalue is $0.9299$, which is closer to $1$ and hence leads to a slower convergence than in the Caltech network (see Figure \ref{fig: convergence}). The second-largest eigenvalue of $\widetilde{A}$ is related to the Fiedler value of the adjacency matrix $A$. The Fiedler value has a strong influence on the time that it takes for random walks and diffusion processes on a network to converge to a steady state~\cite{masuda2016}.
\add{In simulations on SBM networks, we observe a scale separation of the variance in Figure \ref{fig: convergence}. This may depend both on network community structure and on the WTD.}

%%%%%%%

\subsection{Models of opinion dynamics with heterogeneous WTDs}
\label{sec: numeric-heter}

In this subsection, we discuss the effect of WTD heterogeneity on the memory-dependent opinion models \eqref{eq: model-discrete}. 
We consider (1) an example in which all nodes have the same WTD type but different WTD mean values {and} (2) an example in which different nodes have different types of WTDs. 

In our first example, we consider a heterogeneous exponential model \eqref{eq: exp-xi} in which the nodes have exponential WTDs with different values of $\mu$ in \eqref{eq: Texp}. We examine a scenario in which $90$\% of the nodes have a WTD with a mean of $\mu = 1$ and the remaining $10$\% of the nodes have a mean $\mu$ that we vary. 
This ``90-10 decomposition'', which also has been used in a BCM with heterogeneous node-activity levels \cite{li2022bounded}, is motivated by the so-called ``90-9-1 rule'' of participation inequality.
The 90-9-1 heuristic was proposed as a rule of thumb \cite{nielsen200690} that $90$\% of users consume content online but do not contribute to it, $9$\% of users occasionally contribute content, and $1$\% of users account for most contributions of content. In our example, we consider a 90-10 decomposition of a population for simplicity, but one can also consider other situations (such as a 90-9-1 decomposition).

When the mean $\mu < 1$, the minority nodes represent ``open-minded'' entities that tend to change their opinions more frequently than normal nodes. When $\mu > 1$, these nodes represent ``stubborn'' entities that tend to preserve their opinions by interacting less frequently than normal nodes. 
We interpret nodes with long waiting times as stubborn entities, but similar ideas arise in a variety of contexts. Examples include immune nodes in compartmental models \cite{liu2017analysis}, media nodes in opinion models \cite{brooks2020model}, and (as in the present work) stubborn entities in opinion models \cite{ghaderi2014opinion}.

In particular, we consider the parameter values $\mu = 0.2$, $\mu = 1$, and $\mu = 5$ for the minority nodes, as example situations with open-minded nodes, normal nodes, and stubborn nodes, respectively. We generate one complete graph with random weights (see Figure \ref{fig: graphs}(c)) and examine the above three situations on this graph. All three situations have the same initial opinions, which we draw independently from the uniform distribution on $[0,1]$ for each node. We show the opinion trajectories in Figure \ref{fig: x-active-stubborn}. Open-minded nodes tend to converge quickly to the steady state, whereas stubborn nodes change their opinions slowly and attract the mean opinion towards their opinions.

\begin{figure}[htbp]
    \centering
    \includegraphics[width=0.3\textwidth]{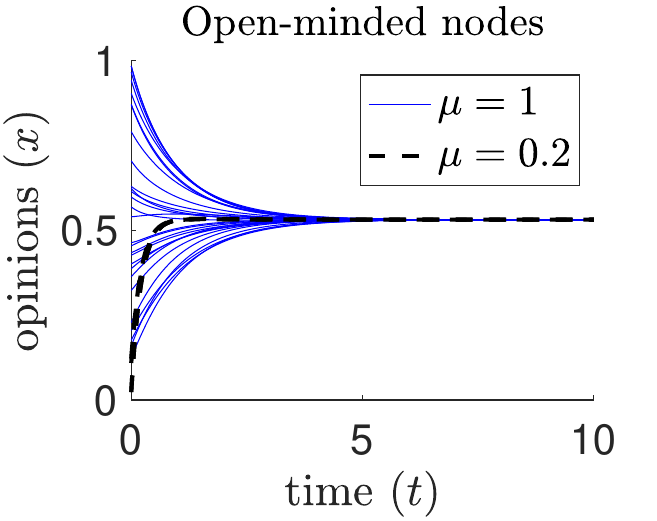}~
    \includegraphics[width=0.3\textwidth]{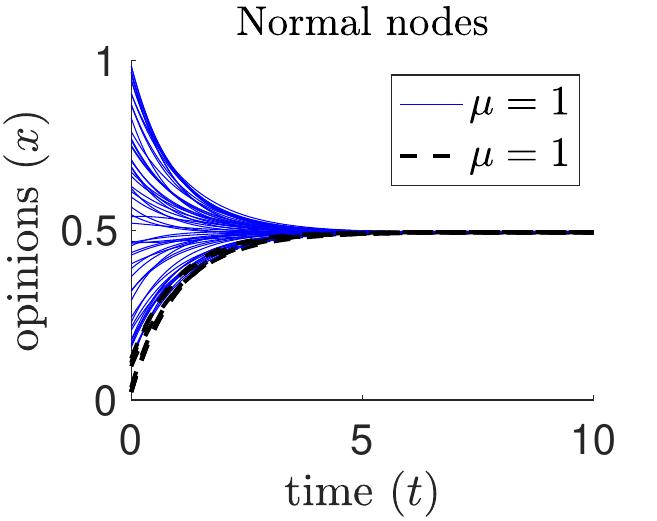}~
    \includegraphics[width=0.3\textwidth]{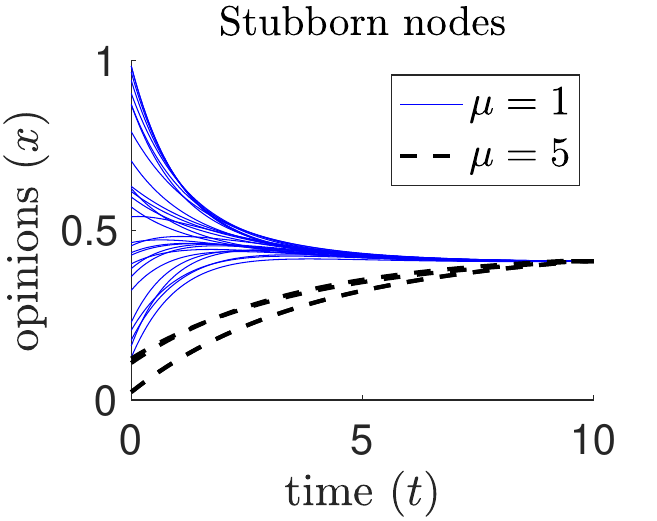}
    \caption{(Color online) Opinion trajectories of the opinion model \eqref{eq: exp-xi} with exponential WTDs on a complete weighted graph (see Figure \ref{fig: graphs}(c)). The graph has $30$ nodes and all nodes have an exponential WTD in \eqref{eq: Texp}. The WTD of $27$ nodes has a mean of $\mu=1$; {the $3$ nodes with the smallest initial opinion values} have means of (left) $\mu = 0.2$, (center) $\mu = 1$, and (right) $\mu = 5$. The steady-state mean opinions are (left) $0.5407$, (center) $0.5015$, and (right) $0.4096$.
    }
    \label{fig: x-active-stubborn}
\end{figure}

In our second example, we examine the effect of WTD heterogeneity on the steady-state opinion clusters when nodes have different types of WTDs. 
We consider a 90-10 decomposition and an 80-20 decomposition of the node WTD types, in which the majority of the nodes with the largest initial opinion values have one WTD (which we call the ``majority WTD'') and that the remaining of the nodes have another WTD (which we call the ``minority WTD''). 
We suppose that the majority WTD is either an exponential distribution \eqref{eq: Texp} with $\mu = 1$ or a Pareto distribution \eqref{eq: Tpareto} with $\alpha = 2$, so that the majority WTD always has a mean of $1$.
We suppose that the minority WTD is one of the six distributions in \eqref{eq: Ts}. We vary the mean of the minority distribution and examine the steady-state clusters of the models \eqref{eq: model-discrete} on a complete graph with random edge weights (see Figure \ref{fig: graphs}(c)).

We plot the steady-state clusters in Figure \ref{fig: heter-clusters}. For each steady-state opinion in Figure \ref{fig: heter-clusters}, we examine $100$ realizations and compute the mean steady-state opinions of the nodes across the realizations. In each realization, we generate a new graph and new initial opinions, which we draw independently from the uniform distribution on $[0,1]$.
We terminate each simulation of the models \eqref{eq: model-discrete} once the opinion variance is less than $10^{-7}$, and we assume that the dynamics have reached a steady state.

\begin{figure}[htbp]
    \centering
    \includegraphics[height=0.25\textwidth]{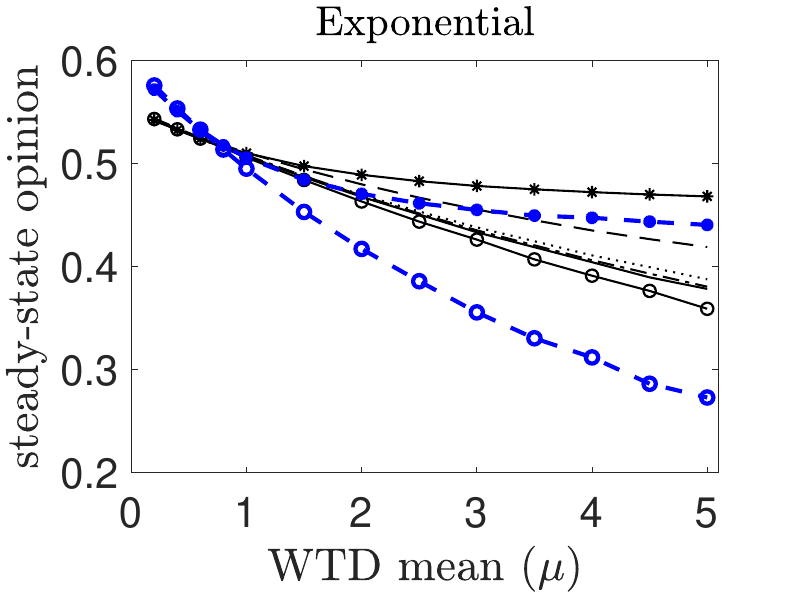}
    \includegraphics[height=0.25\textwidth]{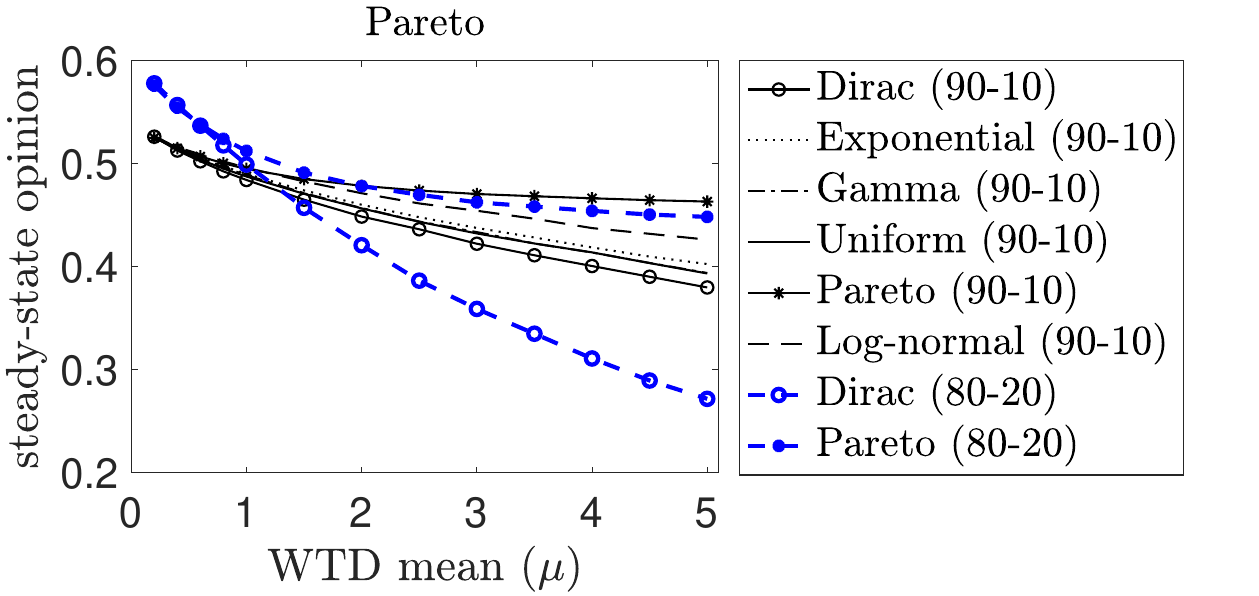}
    \caption{(Color online) Steady-state opinions when a majority of the nodes of a network have (left) an exponential WTD with a mean of $1$ and (right) a Pareto WTD with a mean of $1$. The horizontal axis is the mean value of the minority WTD, and the vertical axis gives the mean steady-state opinions across 100 realizations.
    We use black curves for the mean steady-state opinions in situations with a 90-10 population decomposition (in which 90\% of the nodes have the majority WTD and 10\% of the nodeshave the minority WTD) and blue curves for the mean steady-state opinions in situations with an 80-20 population decomposition. 
    }
    \label{fig: heter-clusters}
\end{figure}

In all of the combinations of WTDs, when the minor nodes are stubborn (i.e., their mean waiting time is larger than that of the normal nodes), the steady-state opinions are smaller than $0.5$, which is the expected value of the initial mean opinion.
The steady-state opinion decreases as we increase the mean waiting time of the stubborn nodes. Of the examined WTDs, stubborn nodes exert the least influence when they have heavy-tailed WTDs and exert the most influence when they have a Dirac delta WTD. The corresponding homogeneous model in which all nodes have the same WTD (either an exponential distribution or a Pareto distribution), which has a mean of $1$, has an expected steady-state mean opinion of $0.5$. When we increase the percentage of special nodes (whether they are stubborn or open-minded) from $10$\% to $20$\%, the heterogeneous models deviate more from their corresponding homogeneous models than when only 10\% of the nodes are special.

%%%%

\subsection{A short remark about polarized and fragmented steady states}

In our numerical simulations of both homogeneous opinion-dynamics models (see Section \ref{sec: numeric-homo}) and heterogeneous opinion dynamics models (see Section \ref{sec: numeric-heter}), we studied examples that converge to a consensus state. However, one can construct examples that converge to a polarized state (which has two distinct opinion clusters) or to a fragmented state (which has three or more distinct opinion clusters) by using graphs whose row-normalized adjacency matrices have more eigenvalues that are equal to $1$.

%%%%%%

\section{Conclusions and discussion} \label{sec: summary}

We proposed a family of memory-dependent models of opinion dynamics that depend on the waiting-time distributions of the nodes of a network.
Our models have continuous-valued opinions and account for memory effects in opinion dynamics. By contrast, to the best of our knowledge, all existing opinion models with continuous-valued opinions yield Markovian descriptions of the time evolution of opinions. 
In our models, the effects of memory emerge naturally from the non-Poisson interevent statistics of the edges of a network.
We illustrate our memory-dependent opinion models using several examples of common WTDs (including Dirac delta distributions, exponential distributions, gamma distributions, and heavy-tailed distributions). When the nodes have a Dirac delta WTD or an exponential WTD, our models have Markovian dynamics and are equivalent to the DeGroot model. When the nodes have a gamma WTD, we obtain a non-Markovian model in which each entity of a network tries to maintain the self-consistency of its opinion as it interacts with other entities. We also approximated heavy-tailed continuous-time WTDs with a sum of Dirac delta functions and derived an associated set of discrete-time opinion models.

We examined convergence to steady states in our models both theoretically and numerically for both homogeneous and heterogeneous scenarios. In homogeneous scenarios, in which all nodes of a network have the same WTD, the time-independent adjacency matrix of the network determines the steady-state opinion. However, the WTD affects the transient dynamics and the rate of the convergence to a steady state. We also observed that models with heavy-tailed WTDs converge more slowly than models with exponentially decaying WTDs. In heterogeneous scenarios, in which nodes have different WTDs (either the same type of WTD with different parameter values or WTDs of different types), ``stubborn'' nodes (i.e., nodes with longer mean waiting times than normal nodes) dominate the overall dynamics by {attracting the mean opinion in a network} towards their opinions.

In the present paper, we proposed non-Markovian models of opinion dynamics with continuous-valued opinions, and we studied some of the properties of these models. There are a variety of interesting ways to extend our investigation. 
We considered WTDs that do not depend on the states (i.e., opinions) of the nodes, and we also assumed that the weights $A_{ij}$ are time-independent. It seems fruitful to examine state-dependent weights (i.e., $A_{ij}(t) = w(x_i(t),x_j(t))$) and other time-dependent weights. For example, one can generalize bounded-confidence models of opinion dynamics \cite{hegselmann2002opinion,deffuant2000mixing} to incorporate memory effects. For example, after an event occurs, suppose that entity $i$ updates its opinion to that of entity $j$ only if their opinions differ from each other by no more than some threshold. 
In our framework, we also assumed that all entities of a network have independent interevent-time statistics, so we did not account for interactions that affect multiple entities at once. One can relax this independence assumption and examine the interdependence that results from coupled {stochastic processes.}
It is also worth considering interdependence between entity opinions and network structure in the form of adaptive (i.e., coevolving) networks \cite{sayama2013modeling}.
In our memory-dependent model that was induced by the gamma WTD, a damping term arose naturally; it promotes self-consistency of the opinion of each entity. It is also worthwhile to explore other memory-dependent models that account for the self-consistency of individual opinions. 
In our study, we examined some scenarios in which the entities of a network follow heterogeneous WTDs. However, we only considered two different WTDs at a time. It is important to investigate more diverse types of heterogeneity, such as systems with many WTDs or with WTDs with randomly-determined parameter values. 

%%%%%%

\section*{Acknowledgements}

We thank Mikko Kivel\"a, Renaud Lambiotte, Naoki Masuda, and our referees for helpful comments. MAP was funded by National Science Foundation (grant 1922952) through the Algorithms for Threat Detection (ATD) program.

%%%%%

\bibliographystyle{siamplain}
\bibliography{bibfile09}

\begin{thebibliography}{10}

\bibitem{asmussen2016laplace}
{\sc S.~Asmussen, J.~L. Jensen, and L.~Rojas-Nandayapa}, {\em On the {Laplace}
  transform of the lognormal distribution}, Methodology and Computing in
  Applied Probability, 18 (2016), pp.~441--458.

\bibitem{bandura1963influence}
{\sc A.~Bandura and F.~J. McDonald}, {\em Influence of social reinforcement and
  the behavior of models in shaping children's moral judgment}, The Journal of
  Abnormal and Social Psychology, 67 (1963), pp.~274--281.

\bibitem{barabasi2005origin}
{\sc A.-L. Barab\'asi}, {\em The origin of bursts and heavy tails in human
  dynamics}, Nature, 435 (2005), pp.~207--211.

\bibitem{baron2022analytical}
{\sc J.~W. Baron, A.~F. Peralta, T.~Galla, and R.~Toral}, {\em Analytical and
  numerical treatment of continuous ageing in the voter model}, Entropy, 24
  (2022), p.~1331.

\bibitem{brooks2020model}
{\sc H.~Z. Brooks and M.~A. Porter}, {\em A model for the influence of media on
  the ideology of content in online social networks}, Physical Review Research,
  2 (2020), 023041.

\bibitem{centola2010spread}
{\sc D.~Centola}, {\em The spread of behavior in an online social network
  experiment}, Science, 329 (2010), pp.~1194--1197.

\bibitem{chen2020non}
{\sc H.~Chen, S.~Wang, C.~Shen, H.~Zhang, and G.~Bianconi}, {\em
  {Non-Markovian} majority-vote model}, Physical Review E, 102 (2020), 062311.

\bibitem{chu2022density}
{\sc W.~Chu and M.~A. Porter}, {\em A density description of a
  bounded-confidence model of opinion dynamics on hypergraphs}, arXiv preprint
  arXiv:2203.12189,  (2022).

\bibitem{deffuant2000mixing}
{\sc G.~Deffuant, D.~Neau, F.~Amblard, and G.~Weisbuch}, {\em Mixing beliefs
  among interacting agents}, Advances in Complex Systems, 3 (2000), pp.~87--98.

\bibitem{degroot1974reaching}
{\sc M.~H. DeGroot}, {\em Reaching a consensus}, Journal of the American
  Statistical Association, 69 (1974), pp.~118--121.

\bibitem{delvenne2015}
{\sc J.-C. Delvenne, R.~Lambiotte, and L.~E.~C. Rocha}, {\em Diffusion on
  networked systems is a question of time or structure}, Nature Communications,
  6 (2015), 7366.

\bibitem{eckmann2004entropy}
{\sc J.-P. Eckmann, E.~Moses, and D.~Sergi}, {\em Entropy of dialogues creates
  coherent structures in e-mail traffic}, Proceedings of the National Academy
  of Sciences of the United States of America, 101 (2004), pp.~14333--14337.

\bibitem{feng2019equivalence}
{\sc M.~Feng, S.-M. Cai, M.~Tang, and Y.-C. Lai}, {\em Equivalence and its
  invalidation between {non-Markovian} and {Markovian} spreading dynamics on
  complex networks}, Nature Communications, 10 (2019), p.~3748.

\bibitem{fernandez2011update}
{\sc J.~Fern{\'a}ndez-Gracia, V.~M. Egu{\'\i}luz, and M.~San~Miguel}, {\em
  Update rules and interevent time distributions: {S}low ordering versus no
  ordering in the voter model}, Physical Review E, 84 (2011), 015103.

\bibitem{friedkin1990social}
{\sc N.~E. Friedkin and E.~C. Johnsen}, {\em Social influence and opinions},
  Journal of Mathematical Sociology, 15 (1990), pp.~193--206.

\bibitem{friedkin2011social}
{\sc N.~E. Friedkin and E.~C. Johnsen}, {\em Social Influence Network Theory:
  {A} Sociological Examination of Small Group Dynamics}, Cambridge University
  Press, Cambridge, UK, 2011.

\bibitem{ghaderi2014opinion}
{\sc J.~Ghaderi and R.~Srikant}, {\em Opinion dynamics in social networks with
  stubborn agents: {E}quilibrium and convergence rate}, Automatica, 50 (2014),
  pp.~3209--3215.

\bibitem{gleeson2014}
{\sc J.~P. Gleeson, D.~Cellai, J.-P. Onnela, M.~A. Porter, and
  F.~Reed-Tsochas}, {\em A simple generative model of collective online
  behavior}, Proceedings of the National Academy of Sciences of the United
  States of America, 111 (2014), pp.~10411--10415.

\bibitem{guttal2010social}
{\sc V.~Guttal and I.~D. Couzin}, {\em Social interactions, information use,
  and the evolution of collective migration}, Proceedings of the National
  Academy of Sciences of the United States of America, 107 (2010),
  pp.~16172--16177.

\bibitem{hallas1997waiting}
{\sc J.~Hallas, D.~Gaist, and L.~Bjerrum}, {\em The waiting time distribution
  as a graphical approach to epidemiologic measures of drug utilization},
  Epidemiology,  (1997), pp.~666--670.

\bibitem{hegselmann2002opinion}
{\sc R.~Hegselmann and U.~Krause}, {\em Opinion dynamics and bounded confidence
  models, analysis, and simulation}, Journal of Artificial Societies and Social
  Simulation, 5 (2002), 3.

\bibitem{hickok2022bounded}
{\sc A.~Hickok, Y.~Kureh, H.~Z. Brooks, M.~Feng, and M.~A. Porter}, {\em A
  bounded-confidence model of opinion dynamics on hypergraphs}, SIAM Journal on
  Applied Dynamical Systems, 21 (2022), pp.~1--32.

\bibitem{hoffmann2012generalized}
{\sc T.~Hoffmann, M.~A. Porter, and R.~Lambiotte}, {\em Generalized master
  equations for {non-Poisson} dynamics on networks}, Physical Review E, 86
  (2012), 046102.

\bibitem{iribarren2009impact}
{\sc J.~L. Iribarren and E.~Moro}, {\em Impact of human activity patterns on
  the dynamics of information diffusion}, Physical Review Letters, 103 (2009),
  038702.

\bibitem{iribarren2011branching}
{\sc J.~L. Iribarren and E.~Moro}, {\em Branching dynamics of viral information
  spreading}, Physical Review E, 84 (2011), 046116.

\bibitem{karsai2011small}
{\sc M.~Karsai, M.~Kivel\"a, R.~K. Pan, K.~Kaski, J.~Kert\'esz, A.-L.
  Barab\'asi, and J.~Saram\"aki}, {\em Small but slow world: {H}ow network
  topology and burstiness slow down spreading}, Physical Review E, 83 (2011),
  025102.

\bibitem{kiss2017}
{\sc I.~Z. Kiss, J.~C. Miller, and P.~L. Simon}, {\em Mathematics of epidemics
  on networks: {F}rom exact to approximate models}, Springer International
  Publishing, Cham, Switzerland, 2017.

\bibitem{kiss2015generalization}
{\sc I.~Z. Kiss, G.~R{\"o}st, and Z.~Vizi}, {\em Generalization of pairwise
  models to {non-Markovian} epidemics on networks}, Physical Review Letters,
  115 (2015), 078701.

\bibitem{kivela2015estimating}
{\sc M.~Kivel{\"a} and M.~A. Porter}, {\em Estimating interevent time
  distributions from finite observation periods in communication networks},
  Physical Review E, 92 (2015), 052813.

\bibitem{lambiotte2015effect}
{\sc R.~Lambiotte, V.~Salnikov, and M.~Rosvall}, {\em Effect of memory on the
  dynamics of random walks on networks}, Journal of Complex Networks, 3 (2015),
  pp.~177--188.

\bibitem{sune-yy2018}
{\sc S.~Lehmann and Y.-Y. Ahn}, {\em Complex Spreading Phenomena in Social
  Systems: {I}nfluence and Contagion in Real-World Social Networks}, Springer
  International Publishing, Cham, Switzerland, 2018.

\bibitem{li2022bounded}
{\sc G.~J. Li and M.~A. Porter}, {\em A bounded-confidence model of opinion
  dynamics with heterogeneous node-activity levels}, arXiv preprint
  arXiv:2206.09490,  (2022).

\bibitem{liu2017analysis}
{\sc Q.~Liu, T.~Li, and M.~Sun}, {\em The analysis of an {SEIR} rumor
  propagation model on heterogeneous network}, Physica A: Statistical Mechanics
  and its Applications, 469 (2017), pp.~372--380.

\bibitem{lorenz2005stabilization}
{\sc J.~Lorenz}, {\em A stabilization theorem for dynamics of continuous
  opinions}, Physica A: Statistical Mechanics and its Applications, 355 (2005),
  pp.~217--223.

\bibitem{mahmud2016poisson}
{\sc T.~Mahmud, M.~Hasan, A.~Chakraborty, and A.~K. Roy-Chowdhury}, {\em A
  {Poisson} process model for activity forecasting}, in 2016 IEEE International
  Conference on Image Processing (ICIP), Institute of Electrical and
  Electronics Engineers, 2016, pp.~3339--3343.

\bibitem{masuda2016guide}
{\sc N.~Masuda and R.~Lambiotte}, {\em A Guide to Temporal Networks}, World
  Scientific, Singapore, 2016.

\bibitem{masuda2009evolutionary}
{\sc N.~Masuda and H.~Ohtsuki}, {\em Evolutionary dynamics and fixation
  probabilities in directed networks}, New Journal of Physics, 11 (2009),
  033012.

\bibitem{masuda2016}
{\sc N.~Masuda, M.~A. Porter, and R.~Lambiotte}, {\em Random walks and
  diffusion on networks}, Physics Reports, 716--717 (2017), pp.~1--58.

\bibitem{masuda2018gillespie}
{\sc N.~Masuda and L.~E. Rocha}, {\em A {Gillespie} algorithm for
  {non-Markovian} stochastic processes}, Siam Review, 60 (2018), pp.~95--115.

\bibitem{medhi1975waiting}
{\sc J.~Medhi}, {\em Waiting time distribution in a {P}oisson queue with a
  general bulk service rule}, Management Science, 21 (1975), pp.~777--782.

\bibitem{meng2018opinion}
{\sc X.~F. Meng, R.~A. Van~Gorder, and M.~A. Porter}, {\em Opinion formation
  and distribution in a bounded-confidence model on various networks}, Physical
  Review E, 97 (2018), 022312.

\bibitem{newman2018networks}
{\sc M.~E.~J. Newman}, {\em Networks}, Oxford University Press, Oxford, UK,
  second~ed., 2018.

\bibitem{nielsen200690}
{\sc J.~Nielsen}, {\em The 90-9-1 rule for participation inequality in social
  media and online communities}, Nielsen Norman Group,  (2006).
\newblock Available at
  \url{https://www.nngroup.com/articles/participation-inequality/} (accessed 19
  August 2022).

\bibitem{noor2020}
{\sc H.~Noorazar, K.~R. Vixie, A.~Talebanpour, and Y.~Hu}, {\em From classical
  to modern opinion dynamics}, International Journal of Modern Physics C, 31
  (2020), 2050101.

\bibitem{porter2016}
{\sc M.~A. Porter and J.~P. Gleeson}, {\em Dynamical systems on networks: {A}
  tutorial}, vol.~4 of Frontiers in Applied Dynamical Systems: Reviews and
  Tutorials, Springer International Publishing, Cham, Switzerland, 2016.

\bibitem{red2011comparing}
{\sc V.~Red, E.~D. Kelsic, P.~J. Mucha, and M.~A. Porter}, {\em Comparing
  community structure to characteristics in online collegiate social networks},
  SIAM Review, 53 (2011), pp.~526--543.

\bibitem{sayama2013modeling}
{\sc H.~Sayama, I.~Pestov, J.~Schmidt, B.~J. Bush, C.~Wong, J.~Yamanoi, and
  T.~Gross}, {\em Modeling complex systems with adaptive networks}, Computers
  \& Mathematics with Applications, 65 (2013), pp.~1645--1664.

\bibitem{scholtes2014causality}
{\sc I.~Scholtes, N.~Wider, R.~Pfitzner, A.~Garas, C.~J. Tessone, and
  F.~Schweitzer}, {\em Causality-driven slow-down and speed-up of diffusion in
  {non-Markovian} temporal networks}, Nature Communications, 5 (2014), 5024.

\bibitem{sood2008voter}
{\sc V.~Sood, T.~Antal, and S.~Redner}, {\em Voter models on heterogeneous
  networks}, Physical Review E, 77 (2008), 041121.

\bibitem{sood2005voter}
{\sc V.~Sood and S.~Redner}, {\em Voter model on heterogeneous graphs},
  Physical Review Letters, 94 (2005), 178701.

\bibitem{starnini2017equivalence}
{\sc M.~Starnini, J.~P. Gleeson, and M.~Bogu{\~n}{\'a}}, {\em Equivalence
  between {non-Markovian} and {Markovian} dynamics in epidemic spreading
  processes}, Physical Review Letters, 118 (2017), 128301.

\bibitem{sugishita2021opinion}
{\sc K.~Sugishita, M.~A. Porter, M.~Beguerisse-D{\'\i}az, and N.~Masuda}, {\em
  Opinion dynamics on tie-decay networks}, Physical Review Research, 3 (2021),
  023249.

\bibitem{sun1998transit}
{\sc Y.-N. Sun and W.~J. Jusko}, {\em Transit compartments versus gamma
  distribution function to model signal transduction processes in
  pharmacodynamics}, Journal of Pharmaceutical Sciences, 87 (1998),
  pp.~732--737.

\bibitem{takaguchi2011voter}
{\sc T.~Takaguchi and N.~Masuda}, {\em Voter model with {non-Poissonian}
  interevent intervals}, Physical Review E, 84 (2011), 036115.

\bibitem{touati2009origin}
{\sc S.~Touati, M.~Naylor, and I.~G. Main}, {\em Origin and nonuniversality of
  the earthquake interevent time distribution}, Physical Review Letters, 102
  (2009), 168501.

\bibitem{traud2012social}
{\sc A.~L. Traud, P.~J. Mucha, and M.~A. Porter}, {\em Social structure of
  {F}acebook networks}, Physica A: Statistical Mechanics and its Applications,
  391 (2012), pp.~4165--4180.

\bibitem{van2013non}
{\sc P.~Van~Mieghem and R.~Van~de Bovenkamp}, {\em {Non-Markovian} infection
  spread dramatically alters the susceptible--infected--susceptible epidemic
  threshold in networks}, Physical Review Letters, 110 (2013), 108701.

\bibitem{vazquez2007impact}
{\sc A.~Vazquez, B.~Racz, A.~Lukacs, and A.-L. Barabasi}, {\em Impact of
  {non-Poissonian} activity patterns on spreading processes}, Physical Review
  Letters, 98 (2007), 158702.

\bibitem{wang2015dynamics}
{\sc W.~Wang, M.~Tang, H.-F. Zhang, and Y.-C. Lai}, {\em Dynamics of social
  contagions with memory of nonredundant information}, Physical Review E, 92
  (2015), 012820.

\bibitem{wojcik2019sizing}
{\sc S.~Wojcik and A.~Hughes}, {\em Sizing up {Twitter} users}, PEW Research
  Center,  (2019).
\newblock Available at
  \url{https://www.pewresearch.org/internet/2019/04/24/sizing-up-twitter-users/}
  (accessed 19 August 2022).

\end{thebibliography}

\end{document}